%% file: main.tex
\documentclass[acmtog]{acmart}

\AtBeginDocument{%
  }
\input{macros}

\usepackage{math}

\copyrightyear{2026}
\acmYear{2026}
\setcopyright{cc}
\setcctype{by}
\acmConference[SIGGRAPH Conference Papers '26]{Special Interest Group on Computer Graphics and Interactive Techniques Conference Conference Papers}{July 19--23, 2026}{Los Angeles, CA, USA}
\acmBooktitle{Special Interest Group on Computer Graphics and Interactive Techniques Conference Conference Papers (SIGGRAPH Conference Papers '26), July 19--23, 2026, Los Angeles, CA, USA}
\acmDOI{10.1145/3799902.3811195}
\acmISBN{979-8-4007-2554-8/2026/07}

\acmSubmissionID{1165}

\begin{document}

\title[MeshFlow]{MeshFlow: Mesh Generation with Equivariant Flow Matching}

\author{Qi Sun}
\email{qisun.new@gmail.com}
\affiliation{%
  \institution{City University of Hong Kong}
  \country{Hong Kong}
}

\author{Kiyohiro Nakayama}
\affiliation{%
  \institution{Stanford University}
  \country{USA}}

\author{Jing Nathan Yan}
\affiliation{%
  \institution{Cornell Tech}
  \country{USA}}

\author{Qixing Huang}
\affiliation{%
  \institution{The University of Texas, Austin}
  \country{USA}}
  
\author{Alexander Rush}
\affiliation{%
  \institution{Cornell Tech}
  \country{USA}}

\author{Leonidas Guibas}
\affiliation{%
  \institution{Stanford University}
  \country{USA}}

\author{Gordon Wetzstein}
\affiliation{%
  \institution{Stanford University}
  \country{USA}}

\author{Jing Liao}
\authornote{Corresponding author}
\affiliation{%
  \institution{City University of Hong Kong}
  \country{Hong Kong}}

\author{Guandao Yang}
\affiliation{%
  \institution{The University of Texas, Austin}
  \country{USA}
}


\input{secs/1_abstract}


\input{figs/teaser}



\maketitle

\input{secs/2_intro}
\input{secs/3_relatedwork}
\input{secs/4_method}

\input{secs/5_experiments}
\input{secs/6_conclusions}


\bibliographystyle{ACM-Reference-Format}
\bibliography{ref}

\input{secs/8_appendix}


\end{document}

%% file: macros.tex
\usepackage{xcolor}
\usepackage{enumitem}
\usepackage{xspace}
\usepackage{booktabs}
\usepackage{colortbl}
\usepackage{multirow}
\usepackage{makecell}
\usepackage{cleveref}
\usepackage{lipsum} 
\usepackage{gensymb} 
\usepackage{wrapfig}
\usepackage{algorithm}
\usepackage{algpseudocode}
\usepackage{etoc}

\usepackage[labelsep=period]{caption}
\newcommand{\qs}[1]{\textcolor{black}{#1}}

\definecolor{MyDarkRed}{rgb}{0.46, 0.16, 0.16}
\definecolor{MyDarkBlue}{rgb}{0.16, 0.16, 0.66}



%% file: secs/1_abstract.tex
\begin{abstract}
Meshes are among the most common 3D scene representations, but directly generating meshes is challenging largely because the mesh representation contains many structures, such as permutation invariance of vertices or faces. 
To address this challenge, we present a novel approach that learns to generate triangle meshes represented as triangle soups.
We adopt equivariant optimal-transport flow matching models that respect key symmetries within the triangle soup representation, including permutation invariance among faces and among vertices within each of the faces. 
Toward this goal, we propose a simple yet effective modification to the state-of-the-art Diffusion Transformer architecture, resulting in a scalable network capable of modeling a flow field while being equivariant to the desirable symmetries.
Moreover, we introduce a loss function grounded in optimal transport principles that improves model convergence by eliminating training signals that violate these symmetries.
Our model can achieve performance comparable to state-of-the-art auto-regressive mesh generators while providing about an 18× speedup during inference.  
Project page is at \url{https://qiisun.github.io/MeshFlow}.
\end{abstract}

%% file: figs/teaser.tex
\begin{teaserfigure}
  \includegraphics[width=\textwidth, width=\textwidth]{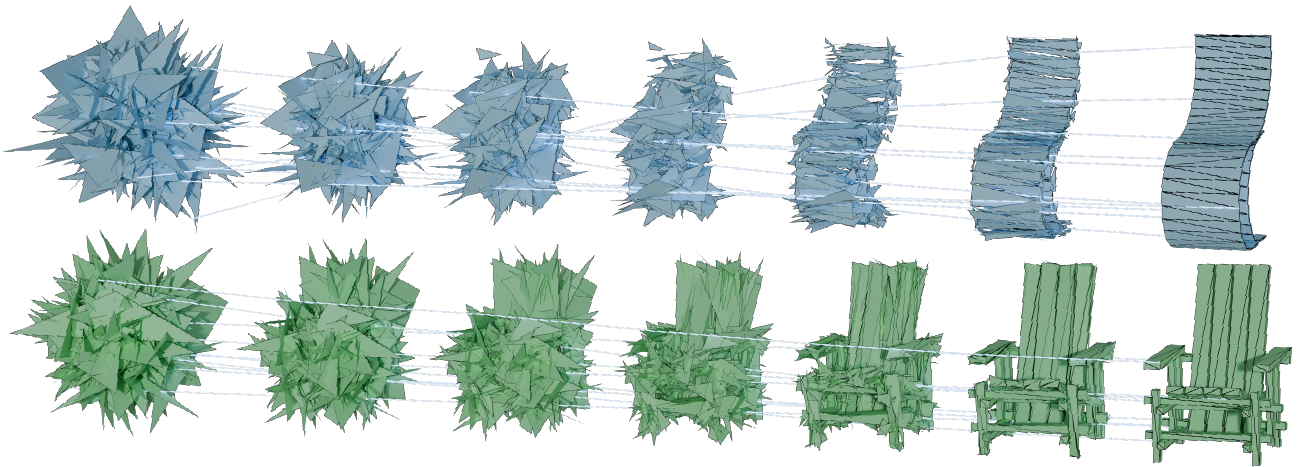}
  \caption{\textbf{MeshFlow} transforms a randomly sampled triangle soup (Left) to a high-quality triangle mesh (Right) in less than 1 second.
  \textbf{MeshFlow} also produces smooth vertex correspondences with minimum crossings, indicated by the lines between triangle soup vertices.
  Each mesh takes less than a second to generate.
  }
  \Description{This is the teaser figure for the article.}
  \label{fig:teaser}
\end{teaserfigure}

%% file: secs/2_intro.tex
\section{Introduction}
\label{sec:intro}

Meshes are among the most widely used representations in computer graphics. 
Many core graphics algorithms in areas such as rendering~\cite{10.5555/3044800}, geometry processing~\cite{Crane:2013:DGP}, and simulation~\cite{kass1993introduction} assume meshes as their main input.
For example, when creating a digital character for animation, artists usually need to carefully mesh regions near joints to enable realistic deformation.
Although various traditional meshing and remeshing tools exist to achieve such goals, they are largely heuristic and still demand careful human intervention to handle different meshing scenarios robustly.
This motivates the study of how to build a mesh generative model that produces meshes that match the quality of artist- and engineer-crafted meshes while flexibly adapting to various conditioning inputs.

A common approach toward this goal is to first generate 3D shapes in alternative representations that are easy to model.
Previous approaches have used point clouds~\cite{zeng2022lion, ShapeGF, spacemesh2024}, implicit functions~\cite{Park_2019_CVPR, occnet, xiang2024structured}, convex primitives~\cite{cvxnet, chen2020bspnet}, or voxels~\cite{ren2024xcube}, as intermediate representations. 
The generated shapes are then converted from the intermediate representation into meshes using algorithms such as Marching Cubes~\cite{marching_cube}.
However, such two-stage pipelines are often limited by the mesh quality of the second stage algorithms, which can create artifacts such as over-tessellated surfaces. 
Consequently, the two-stage method tends to produce meshes that lack the intentional tessellation characteristic of those crafted by human artists and engineers.

An alternative approach is to learn a generative model directly from human-curated mesh datasets.
Recently, several works have taken steps toward this direction by applying autoregressive generative models to serialized mesh representations~\cite{nash2020polygen, siddiqui2024meshgpt}.
These works demonstrate state-of-the-art mesh quality and can create tesselations that are highly similar to those created by human users.
However, these models inherit the fundamental limitations of autoregressive methods when applied to high-dimensional data such as meshes.
They often suffer from slow inference speeds~\cite{lou2023discrete}, difficulty to control~\cite{li2022diffusion}, and error accumulation when generating long sequences~\cite{holtzman2019curious}.

In this work, we propose \textsl{MeshFlow},  a mesh generative model that learns directly from human-created mesh data using a special class of diffusion model - equivariant optimal-transport flow-matching models~\cite{klein2023equivariant}.
Compared to autoregressive models, flow matching has the potential to achieve fast inference speed ~\cite{nie2025largelanguagediffusionmodels,liu2022flow} and can be adapted to take different control signals via techniques such as diffusion posterior sampling~\cite{chung2023diffusion}. 
However, recent exploration in directly applying diffusion models to mesh generation finds it difficult to produce results matching the state-of-the-art auto-regressive mesh generative models.
We hypothesize that a key factor limiting the success of prior diffusion-based methods is their failure to account for the inherent symmetries of faces and vertices in meshes.
To address this, we design a flow-matching model that respects these symmetries by proposing two technical contributions.
First, we present a simple yet effective modification of the Diffusion Transformer architecture~\cite{peebles2023scalablediffusionmodelstransformers}, resulting in a powerful and scalable neural network capable of modeling the velocity field of a triangle soup while maintaining equivariance to its key symmetries.
Second, we introduce a loss function grounded in optimal transport principles to eliminate training signals that violate these symmetries.
Our training objectives enable stable training and faster inference.
Illustrated in~\Cref{fig:teaser}, MeshFlow generates mesh samples from Gaussian noise with a smooth and straight velocity field.

We demonstrate that \textit{MeshFlow} can achieve results \textit{on par in mesh quality} with state-of-the-art mesh generation methods in various ShapeNet categories.
However, \textit{MeshFlow} can produce high-quality meshes in less than a second, which is \textit{18 times} faster than the state-of-the-arts using auto-regressive models. 
These results highlight the potential of our approach to different interactive applications in computer graphics pipelines.

To summarize, the key contributions of our paper include: 
\begin{itemize}
[leftmargin=1.5em]
    \item We propose a novel mesh generation method by applying optimal-transport equivariance flow matching models on triangle soup.
    \item We propose a variant of DiT architecture that respects two types of invariance in triangle soups: 1) invariance to the permutation of faces, and 2) invariance to the cyclic rotation of triangle vertices.
    To improve training convergence, we propose the appropriate optimal transport loss function, which couple each sampled data point with a noise permuted to minimize its distance to the data point.
    \item We demonstrate the effectiveness of our model through experiments on ShapeNet dataset, reaching state-of-the-art mesh generation quality with sub-second sampling time. 
\end{itemize}

%% file: secs/3_relatedwork.tex
\section{Related Work}
\label{sec:related_works}

In this section, we will first discuss three mesh generation approaches: those using intermediate representations, autoregressive models, and diffusion models. 
We will also review recent advances in equivariant optimal-transport flow matching for sets and graphs -- the core technique adapted in this work.

\paragraph{Mesh Generation through Intermediate Representation.}
One common approach is to first generate an intermediate shape representation and apply geometry processing algorithms to transform it into meshes.
For example, previous work~\cite{ShapeGF, pointflow, zeng2022lion} invested in developing a high-quality point cloud generative model. 
In the second stage, they can be turned into meshes using surface reconstruction algorithms such as Poisson Surface Reconstruction~\cite{PSR, peng2021shape}.
Another popular intermediate representation is the voxel. 
Due to its regularity and compatibility with 3D deep neural networks, methods such as~\cite{3dgan,ren2024xcube, lu2024infinicube,brock2016generative,choy20163d,jimenez2016unsupervised,wu2016learning}
first generate voxels and use Marching Cube~\cite{marching_cube} to subsequently turn the voxels into a mesh. 
One drawback of voxel-based generation is its high memory footprint, making high-fidelity 3D generation expensive.
To avoid this issue, previous studies also explored the generation of neural implicit representations~\cite{Park_2019_CVPR, occnet, chen2018implicit_decoder, xiang2024structured,xu2024instantmesh, Liu2023MeshDiffusion, zhang2024clay, zhang20233dshape2vecset}, primitive collections of 2D patches~\cite{groueix2018, yang2025atlas,xu2024brepgen,yan2024object}, BSP-trees~\cite{chen2020bspnet}, or hyper-planes~\cite{cvxnet}.
Although these approaches can circumvent the difficulty of building a generative model of the irregular mesh data, converting their results into meshes remains challenging. 
For example, Marching Cubes can create irregular faces and lose sharp features~\cite{siddiqui2024meshgpt, chen2024meshxl}. 
Furthermore, the tessellation created by the meshing algorithms might be different from and potentially inferior to those created by artists, requiring another round of remeshing for downstream applications. 
In contrast,  MeshFlow outputs meshes directly and is learned from a given mesh dataset.

\paragraph{Auto-regressive Mesh Generation.}
Another popular way to learn both surface and discretization from data is to turn the mesh into a sequence and learn to generate it auto-regressively.
PolyGen~\cite{nash2020polygen} pioneers this idea of learning directly from raw mesh data by introducing two auto-regressive models, one for vertices and the other for edges conditioned on the generated vertices. 
MeshGPT~\cite{siddiqui2024meshgpt} operates on a sequence of latent vectors produced from a graph-based Vector Quantized  VAE~\cite{oord2018neuraldiscreterepresentationlearning,lee2022autoregressive}.
More recent methods~\cite{lionar2025treemeshgpt, chen2024meshxl, chen2024meshanything, tang2024edgerunner,weng2024scaling,chen2024meshanythingv2artistcreatedmesh,wang2025nautilus,hao2024meshtronhighfidelityartistlike3d, weng2024pivotmesh} try to perform auto-regressive generation directly in a single stage, building a network that can output a sequence containing both face and vertex information. 
These methods usually differ in terms of their mesh tokenization scheme and sometimes the network architecture. 
For example, \cite{chen2024meshxl,chen2024meshanything,chen2024meshanythingv2artistcreatedmesh,tang2024edgerunner, wang2025nautilus,lionar2025treemeshgpt, wang2024llamameshunifying3dmesh, wang2025iflame} create more efficient and compact tokenizers, leveraging insights such as the half-edge data structure. 
Meshtron~\cite{hao2024meshtron} proposed \qs{an efficient hourglass transformer} architecture that can process a large number of faces. 
Many approaches have also explored using auto-regressive mesh generative models for downstream tasks such as conditional generation tasks~\cite{gao2024meshart,li2025meshpad, Zhang_2025_ICCV_VertexRegen, fang2025meshllm, lei2025armesh, shen2025flashmeshfasterbetterautoregressive, xu2025meshmosaic} and preference finetuning~\cite{zhao2025deepmeshautoregressiveartistmeshcreation}. 
However, these methods are usually bottlenecked by limitations inherent from auto-regressive generative models, such as their slow inference speed, the difficulty in defining a canonical ordering of the mesh faces, and error accumulation when generating a long sequence. 
Our paper aims to apply flow-matching models to circumvent these limitations, allowing efficient inference.

\paragraph{Diffusion-based Mesh Generation.}
Diffusion models \cite{song2020denoising, sohl2015deep, ho2020denoising} and flow matching \cite{liu2022flow, klein2023equivariant, Esser2024ScalingRF, Lee2024ImprovingTT} both iteratively create and refine all dimensions of the data simultaneously.
They can generate high-dimensional data more efficiently than auto-regressive models and can be used as a plug-and-play prior through different posterior sampling techniques~\cite{chung2023diffusion}.
Several preliminary attempts have been made to apply diffusion-based generation to meshes. 
For example, Polydiffuse \cite{chen2023polydiffuse} uses diffusion models to generate sets of 2D polygonal shapes. SpaceMesh \cite{spacemesh2024} first generates vertices via diffusion and then learns an embedding through a self-supervised loss to recover vertex connectivity. 
\qs{PolyDiff \cite{alliegro2023polydiff} leverages a categorical diffusion model with an architecture similar to UViT~\cite{bao2022all} to produce a quantized triangle soup }.
\qs{A concurrent work, MeshCraft \cite{he2025meshcraft}, applies rectified flow in a latent space ordered by the PolyGen method, built upon a diffusion transformer with RoPE~\cite{su24rope}}.

MeshFlow differs from these works in two ways. 
First, none of them applies diffusion directly to raw mesh data, limiting the generation quality to either the discretization errors~\cite{alliegro2023polydiff} or the quality of the mesh autoencoder~\cite{he2025meshcraft}. 
Second, some methods overlook inherent symmetries in their mesh representations, which limits their efficiency.
In contrast, MeshFlow is a simple yet effective framework for generating triangle meshes in the continuous space by fully leveraging the invariance inherent to meshes represented as triangle soups. 
We take motivation from optimal-transport equivariant diffusion and flow-matching models \cite{pmlr-v108-niu20a, hoogeboom2022equivariantdiffusionmoleculegeneration, jo2022scorebasedgenerativemodelinggraphs,vignac2023digressdiscretedenoisingdiffusion,klein2023equivariant,song2023equivariant, hui2025NotSoOT}, which aim to model a distribution that is invariant to certain symmetries.
However, these methods are mainly focused on the generation of molecules, graphs, and point clouds. Instead, Meshflow is the first method to apply equivariant generation to meshes represented as triangle soup, and we study the important design decisions needed for such non-trivial adaptation.

%% file: secs/4_method.tex

\input{figs/framework}

\section{Preliminaries}

\paragraph{Flow matching.}
Flow Matching (FM)~\cite{liu2022flow, lipman2022flow} produces a Continuous Normalizing Flows~\cite{chen2018neuralode} while avoiding expensive simulation steps typically required in their training.
The core idea is to define a conditional vector field $u_t(\cdot|x_1)$ and a corresponding path $p_t(\cdot|x_1)$ that deterministically transform samples from a prior distribution $q_0$ (e.g., Gaussian) into a Dirac delta distribution centered at a target data point $x_1$ when $t=1$.
Flow matching demonstrates that a neural network $v_{\theta,t}$, which models this conditional vector field, can be trained efficiently using a straightforward Conditional Flow Matching (CFM) objective:
\begin{equation} \label{eq:cfm_from_snippet} 
\mathcal{L}_{\text{CFM}} = \mathbb{E}_{t,q_1(x_1),q_0(x_0)} [\|v_{\theta,t}(x_t) - u_t(x_t|x_1)\|^2].
\end{equation}
A prevalent and simple choice for this conditional setup involves defining the target vector field as $u_t(x|x_1) := x_1 - x_0$, where $x_0 \sim q_0$. This target field corresponds to paths
\begin{equation}
    x_t = (1-t)x_0 + tx_1,
\end{equation}
which are linear interpolations between the noise sample $x_0$ and the data sample $x_1$.
The standard CFM objective in Eq.~(\ref{eq:cfm_from_snippet}) often relies on an \textit{independent coupling} where the noise/data sample pair $(x_0, x_1)$ is sampled from independent distributions $q_0(x_0)$ and $q_1(x_1)$. 
Some previous work~\cite{tong2024improving, pooladian2023multisample} shows that 
an OT map $\pi$ that minimizes $\int C(x_0, x_1)^2 \pi(x_0, x_1)dx_0dx_1$ with the cost function $C(x_0, x_1) = ||x_0 -x_1||$, is a good choice for data coupling since it can lead to straighter trajectories.
However, obtaining such optimal transport maps is often intractable.

\paragraph{Data symmetry and Coupling.}
Fortunately, much data exhibit inherent symmetries given by a certain group $G$.
The probability distribution of such data then becomes invariant w.r.t. actions by $G$. E.g., $P(x) = P(g \cdot x)$ for all data points $x$ and $g\in G$.
Such data symmetry provides a good prior to reduce optimal transport costs.
Previous equivariant OT flow matching works~\cite{song2023equivariant, klein2023equivariant} exploit such prior to generate elements invariant to actions by certain groups, such as permutations, rotations, and translations. 
Specifically, they propose to define the cost function $C(x_0, x_1)$ with one that accounts for these group elements: 
\begin{equation}
    C(x_0, x_1) = \min_{g\in G} || x_0 - g\cdot x_1||^2.
    \label{eq:cost}
\end{equation}
This approach significantly reduces the OT distance even with a small batch size, demonstrating success in modeling structured data such as molecules and point-clouds.
Our work identifies the key symmetries in triangle soups and proposes an efficient approximation to leverage such symmetry to reduce OT flow matching costs.

\paragraph{Equivariant architecture.}
Another important component to ensure an invariant probability distribution is to make use of an equivariant neural network to parameterize $v_\theta$:
prior works have shown that if $v_\theta$ is equivariant to the group $G$, \textit{i.e.}, $v_\theta(g\cdot x, t) = g\cdot v_\theta(x, t)$, and $p_0$ is invariant to $G$, then the probability $p_t$ induced by $v_t$ applied to $p_0$ is also invariant to $G$~\cite{satorras2021n, ballerin2025so, lawrence2025improving}.
\qs{Inspired by early exploration on equivariant architecture network designed for geometric data~\cite{qi2017pointnet, fuchs2020se, satorras2021n}}, our work proposes a simple and effective neural network architecture that is equivariant to the key symmetries in triangle soups.

\section{Method}
We aim to train an unconditional generative model to generate meshes $\mathcal{M}= \{\bm V, \bm F\}$.
Departing from graph-based mesh representations, we adopt a triangle soup representation that shares some symmetry, which will be introduced in Sec.~\ref{subsec:tri}.
Fig.~\ref{fig:framework} shows the framework of our proposed methods: in Sec.~\ref{subsec:train}, we first define a metric between triangle soups, then find the noise-data OT coupling w.r.t. to the metric.
In Sec.~\ref{subsec:eq-arch}, we introduce the equivariant rectified flow network architecture specifically designed for triangle soups.

\subsection{Triangle Soup Representation}
\label{subsec:tri}
A triangle soup is composed of a set of triangle faces $x\in \mathbb{R}^{N\times3\times3} = \{\mathbf f_1, \mathbf f_2, \cdots, \mathbf f_N\}$, where $N$ denotes the number of triangular faces. 
And each face $\mathbf f_i = \{\mathbf v_i^j\}_{j=1}^3$ comprises three unordered vertices. Note that we do not model the orientation of each face, since we do not explicitly model topology.
This set-based representation avoids imposing sequential dependencies, unlike autoregressive approaches to mesh generation.

\paragraph{Data symmetry.}
The triangle soup exhibits two levels of permutation symmetries: 
(1) Face-level: The $N$ triangles comprising the mesh can be arbitrarily permuted without altering the underlying geometry. This corresponds to the symmetric group $S_N$;
(2) Vertex-level: Within each triangle, the order of its three vertices is irrelevant because a triangle soup does not contain connectivity information. Thus, it is invariant to the symmetric group $S_3$.
Together, the two levels of permutation invariance form a subgroup of $S_{3N}$, which we will denote as $G$. Mathematically, $G = S_3 \wr_N S_N$ is the wreath product between $S_3$ and $S_N$. Its group action on a set of $3N$ elements $x = \set{\mathbf f_1, \mathbf f_2, \cdots, \mathbf f_N}$ is given by
\begin{equation*}
    ((\sigma_i)_{i=1}^N, \rho)\cdot x = \set{\mathbf{f}^{\sigma_1}_{\rho(1)}, \mathbf{f}^{\sigma_2}_{\rho(2)}, \cdots, \mathbf{f}^{\sigma_N}_{\rho(N)}},
\end{equation*}
with each $\mathbf{f}^{\sigma_i}_{\rho(i)} = \set{\mathbf{v}^{\sigma(j)}_{\rho(i)}}_{j=1}^3$,
for any $((\sigma_i)_{i=1}^N, \rho)\in G$, $\sigma_i\in S_3$ and $\rho \in S_N$. 

\paragraph{Discussion.}
The triangle soup representation is also employed by PolyDiff~\cite{alliegro2023polydiff} due to its ability to comprehensively encode mesh geometry and its seamless extensibility to quadrangular or general polygon meshes. This representation forms the basis of our work.
However, whereas PolyDiff uses quantized categories for triangle representation, our method directly models the continuous spatial structure and explicitly addresses the aforementioned data symmetries.

\subsection{Equivariant Architecture}
\label{subsec:eq-arch}
Equivariant flow matching requires an equivariant velocity predictor. 
While transformer architecture is permutation equivariant by default, recent work usually finds it necessary to add positional encoding to achieve good performance~\cite{peebles2023scalablediffusionmodelstransformers}. 
However, positional encoding breaks the equivariance of transformers.
In this section, we will introduce a simple modification of the diffusion transformer (DiT)~\cite{peebles2023scalablediffusionmodelstransformers} architecture, maintaining the necessary equivariance to group $G$ as defined in Sec.~\ref{subsec:tri} and good computational efficiency and performance.
The right side of \cref{fig:framework} provides an illustration of our network -- it consists of a vertex positional embedder, a series of equivariant DiT blocks, and an output layer.

\paragraph{Vertex embedder.}
Given a triangle soup $x\in \R^{N\times 3\times 3}$, the vertex embedder encodes each vertex coordinate $p := \bm{v}_i^j$ using sinusoidal positional encoding $\gamma(p)$, similar to NeRF~\cite{mildenhall2020nerf}:
\begin{equation}
    \gamma(p) = (\sin (2^0\pi p), \cos (2^0\pi p), \cdots, \sin (2^{L-1}\pi p), \cos (2^{L-1}\pi p)),
\end{equation}
where $L$ is the number of frequencies.
An MLP then maps these vertex embeddings to a hidden dimension.

\input{figs/eq-dit}

\paragraph{Equivariant DiT block.}
The core of our model is the Equivariant Diffusion Transformer (DiT) block (\Cref{fig:eq-dit})
We would like to design an architecture that is equivariant to $G$ while remaining computationally efficient.
Note that the original DiT block, when applied to all $3n$ vertices, is already equivariant to the larger permutation group $S_{3N}$.
However, the computational complexity in the self-attention layer scales quadratically to the number of tokens, making it computationally expensive to process all vertices in the triangle soup. Furthermore, since $S_{3N}$ contains $G$, being equivariant to $S_{3N}$ will make the network less expressive, as different orderings of \textit{non-equivalent} triangle soups will obtain the same features. 
To tackle the issues mentioned above, we propose to modify the DiT block to be only equivariant to $G$. 

Our key idea is to perform self-attention on faces to aggregate invariant global information for each face vertex to achieve an expressive equivariant DiT block. 
Specifically, we first aggregate the input vertex features $\{\mathbf{v}^0_i, \mathbf{v}^1_i, \mathbf{v}^2_i\}_{i=1}^N$ for each face into face features $\{\mathbf{f}_i\}_{i=1}^N$ using average-pooling, which are processed by a self-attention layer without positional encoding.
Then, the output of the self-attention layer will be duplicated and added to the vertex embeddings to preserve the permutation equivariance of vertices within a triangle face.
Finally, a two-layer feed-forward network (FFN) processes each vertex feature. 
Similar to the original DiT layer, we apply an Adaptive Layer Normalization (adaLN) layer before both the self-attention layer and the FFN layer.
The adaLN layer modulates the features based on both the timestamp $t$ and the number of faces $|F|$.
Although directly applying the DiT block to face features is more efficient 
(i.e., $N^2D+11ND^2$ MACs), it leads to significantly worse performance (See \Cref{sec:ablation}) because vertex-level symmetries are not modeled.

\paragraph{Face conditioning.}
Meshes generated with different face budgets can exhibit different geometric characteristics. 
With a limited face budget, models tend to approximate flat surfaces with larger triangles, whereas generous budgets usually encourage more curved surfaces with smaller triangles. 
Although in principle one could infer the target face count from the number of tokens, the softmax normalization in the attention layer makes it difficult to recover this information from the DiT block without positional encoding~\cite{meta2025llama,kocher2025nope}. 
To address this, we explicitly condition the network on the desired number of faces by embedding this scalar into the Adaptive LayerNorm (adaLN) parameters, ensuring that the model can adapt its generation strategy to any specified face budget.
We create an embedding vector for all $B$ consecutive face numbers, which is added together with the timestep embedding to create the conditioning vector for the adaLN layer in the equivariant DiT block.


\subsection{Symmetry-aware Training Objectives}
\label{subsec:train}
To fully exploit the two-level permutation invariance of triangle soup, we follow prior works on equivariant flow matching~\cite{song2023equivariant, klein2023equivariant} and build the data coupling between noise and data that respect the group $G$.
Building on Eq.~\ref{eq:cost}, we define the cost between noise $x_0$ and triangle soup $x_1$ as the minimal squared distance $\ell_2$ across their orbits under $G$:
\begin{equation}
\label{eq:our_cost} 
C(x_0, x_1) = \min_{((\sigma_i)_{i=1}^N, \rho) \in G}\norm{x_1 - ((\sigma_i)_{i=1}^N, \rho)\cdot x_0}^2.
\end{equation}
In this paper, we restrict ourselves to finding the group action that minimizes $C(x_0, x_1)$, rather than solving for an optimal coupling between different $(x_0,x_1)$. 
Previous work has shown that this approach is computationally efficient and performant~\cite{hui2025NotSoOT}.

\input{figs/illustrate-ot}
\paragraph{Nested coupling.} 
To find the action that minimizes the initial coupling $(x_0, x_1)$ with respect to Eq.~\ref{eq:our_cost}, we first construct the pairwise face cost matrix:
\begin{equation}
    \sigma_{kl}= \operatorname{arg}\min_{\sigma\in S_3}\norm{f_1^l - \sigma \cdot f_0^k}^2
     \text{ and }
     M_{kl}= \norm{f_1^l - \sigma_{kl} \cdot f_0^k}^2,
    \label{eq:cost_face}
\end{equation}
where $f_1^l, f_0^k\in \R^{3\times 3}$ are the $l$-th and $k$-th faces of the clean and noise triangle soup, respectively. 
Once we compute $\mathbf{M}$, Hungarian algorithm~\cite{Kuhn1955} is used to solve the linear assignment problem that yields the optimal face permutation 
$$\rho^* = \arg\min_{\rho\in S_N}\sum_{i=1}^N M_{i, \rho(i)}.$$
This permutation establishes a bijective map $\rho^*$ where $f_0^k$ is matched to $f_1^{\sigma^* (k)}$ such that their L2 distance is minimal modulo permutation of their vertices.
After computing the face-level correspondence $\rho^*\in S_N$ as above, we retrieve vertex-to-vertex correspondences $\sigma^*_i$ for each coupled face pair $(f_0^i, f_1^{\rho^*(i)})$ from Eq.~\ref{eq:cost_face}: $\sigma_i^* = \sigma_{i\phi^*(i)}$.
With this, we obtain that $((\sigma^*_i)_{i=1}^N, \rho^*)\in G$ is the desired coupling between $x_0$ and $x_1$.

Figure~\ref{fig:ot} illustrates the nested coupling strategy for 2D triangles. 
Gray triangles depict noisy faces, while blue triangles represent clean faces (e.g., derived from a Delaunay triangulation). 
Our nested coupling, compared to \textit{naive face coupling} derived from the cost function in Eq.~\ref{eq:cost}, significantly reduces path crossings in the visualized transport plan, suggesting a more stable and coherent generative process. 

Our final objective applies the standard Conditional Flow Matching (CFM) objective while utilizing the OT coupled noise $\tilde{x}_0$ as defined in \cref{eq:our_cost}:
\begin{equation}
    \mathcal{L}(\theta) = \mathbb{E}_{t\sim\mathcal{U}[0,1],x_0\sim q(x_0),x_1\sim q_1(x)}[\norm{v_\theta(x_t, t; c) - (x_1 - \tilde{x}_0)}^2],
\end{equation}
where $x_t = t \cdot x_1 + (1-t) \cdot \tilde{x}_0$ is the linear interpolation between clean triangle soup $x_1$ and the noise after applying the optimal group action to minimize the distance in \cref{eq:cost_face}. Figure~\ref{fig:ot-inf-vis} shows a toy equivariant flow matching model trained on a set of face vertices. Notice that our nested coupling indeed achieves a straighter integration path compared to the independent coupling. This confirms that the nested coupling achieves a substantially lower optimal transport cost.

\begin{wrapfigure}[9]{l}{0.15\textwidth}
    \centering
    \includegraphics[width=0.15\textwidth]{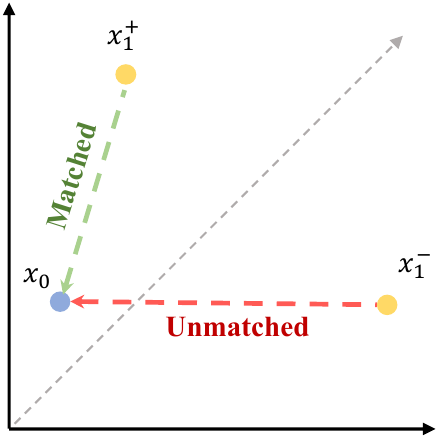}
\end{wrapfigure}
The left inset illustrates how our coupling works in two dimensions. 
Imagine learning a flow-matching model for a permutation-invariant set of two scalars, represented as a 2D point $x_1\in\mathbb{R}^2$.
Under actions given by $S_2$ (swapping the two coordinates), each $x_1^+$ is identified with $x_1^-$.
When we couple $x_0$ to $x_1 = \set{x_1^+, x^-_1}$, our strategy picks the permutation that minimizes the coupling cost and automatically rejects the pairing of $(x_1^-, x_0)$.
In practice, this means that any velocity supervision would only match the data $x_0$ to noise at the same side of the symmetry boundary.
Consequently, the learned velocity field $v_\theta$ consistently predicts directions within the original ordering, rather than averaging over the symmetric configurations.

\begin{figure*}[ht]
    \centering
\includegraphics[width=0.95\linewidth]{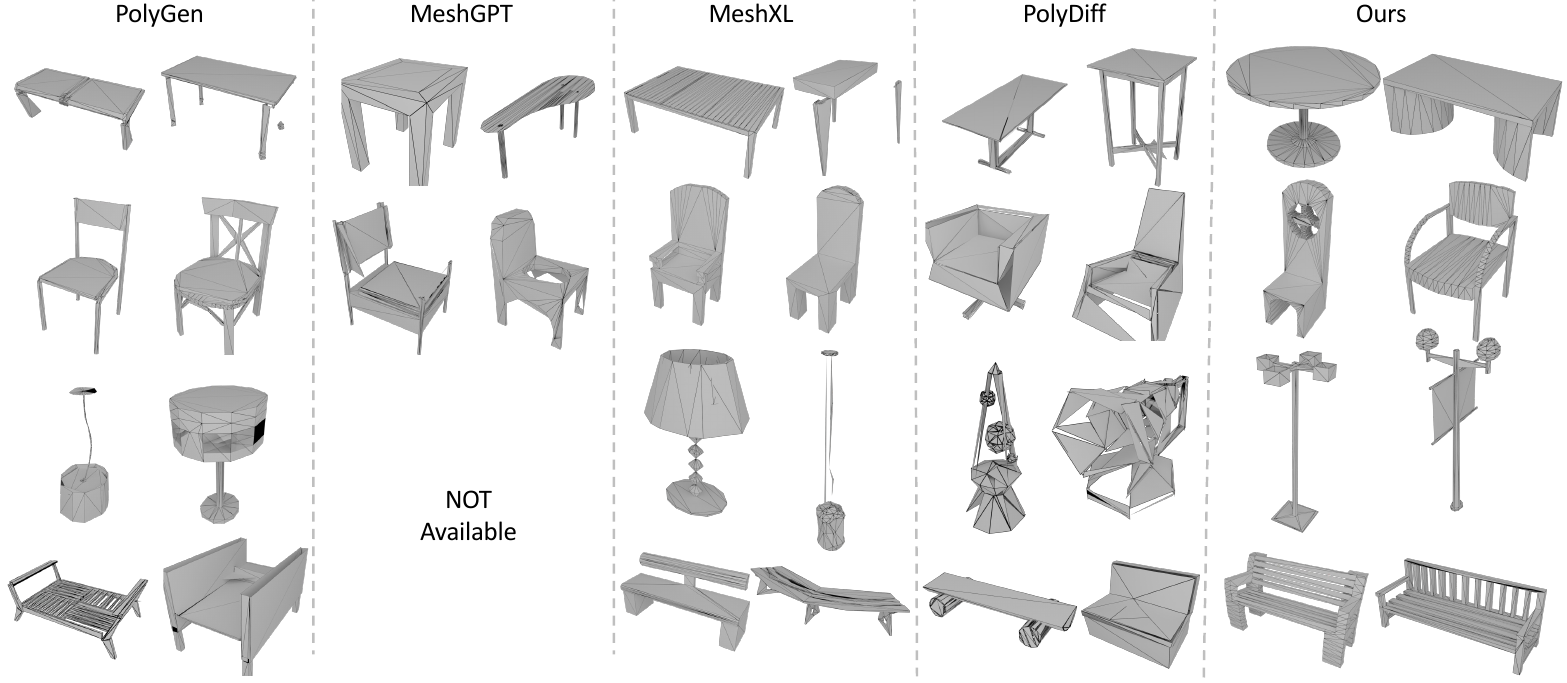}
    \caption{\textbf{Qualitative comparison with the-state-of-the-art methods.} }
    \label{fig:comparison-sota}
\end{figure*}

\begin{figure*}
    \centering
    \includegraphics[width=0.95\linewidth]{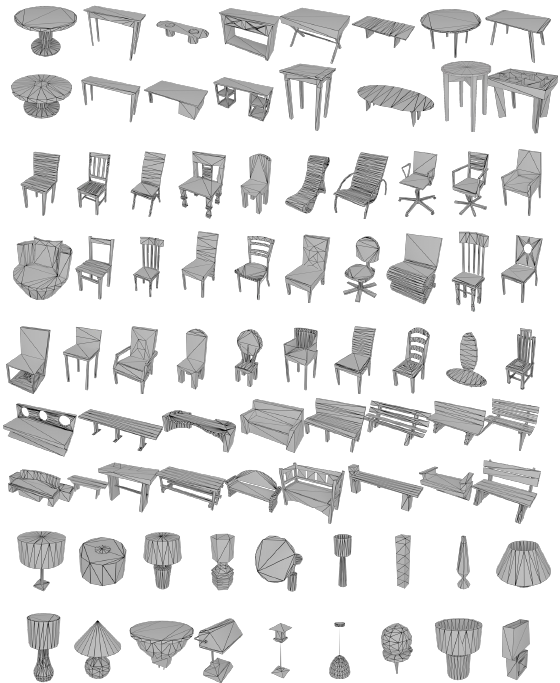}
    \caption{\textbf{Gallery of our generated meshes.}}
    \label{fig:gallery}
\end{figure*}

\subsection{Post-processing}

\begin{wrapfigure}[15]{r}{0.25\textwidth}
    \centering
    \includegraphics[width=0.25\textwidth]{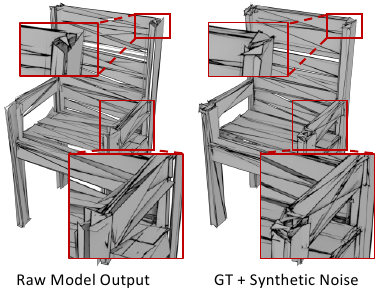}
    \label{fig:denoiser-motivate}
    \caption{Left: generated outputs. Right: closest ground-truth mesh with synthetic Gaussian noise.}
\end{wrapfigure}
To produce a mesh using our model, we follow prior works~\cite{Esser2024ScalingRF} to use the first-order Euler method with $50$ sampling steps.
In contrast to autoregressive methods that predict logits for quantized coordinates, our continuous diffusion framework generates a triangle soup with vertices in a continuous domain. 
This raw output often contains clusters of near-coincident vertices (\cref{fig:denoiser-motivate}).
To produce a good output, we apply a two-step post-processing.

The first post-processing step involves a neural network denoiser. 
We observe that the generated mesh resembles a ground-truth mesh with a small amount of added Gaussian noise (see the right inset).
We hypothesize that if we could train a denoiser to recover the ground-truth mesh from the synthetically noised mesh, then our denoiser could potentially be applicable to also the generated mesh from our model.
Toward this end, we train a mesh denoiser that consumes meshes with a fixed level of Gaussian noise $\eta\epsilon$ added to the ground truth meshes $x$. 
The denoiser is trained to optimize a simple $L_2$ reconstruction loss between the ground truth and the output without any optimal-transport mechanism to permute the face or vertices:
\begin{equation}
    \mathcal L_{\text{denoiser}} = ||f_\theta(x+ \eta \cdot \epsilon) - x||_2^2.
\end{equation}
The denoiser used the same structure as the EquiDiT block without the AdaLN layers.

In the second step, we first apply a thresholding-based clustering algorithm to consolidate these closely located vertices to obtain a set of unique vertices.
All vertices that are within $0.015$ of each other will merge.
We then remove identical faces in the second step.

%% file: figs/framework.tex
\begin{figure*}
    \centering
    \includegraphics[width=0.9\linewidth]{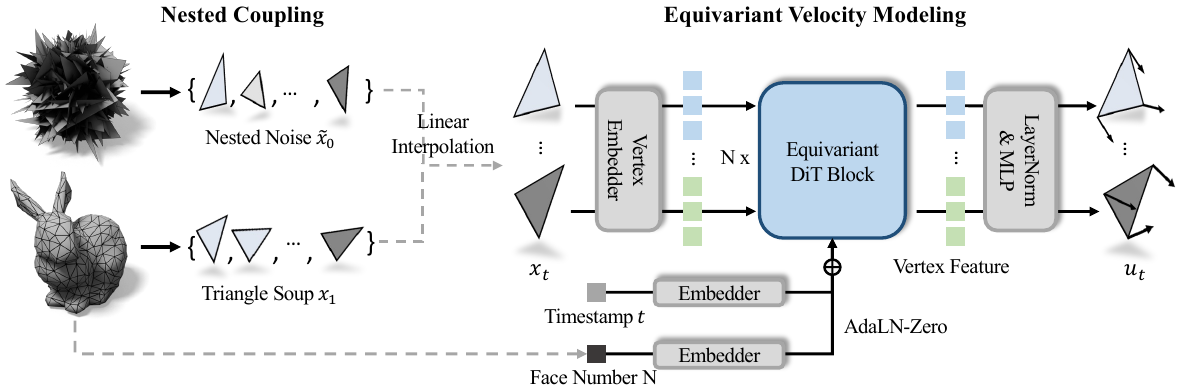}
    \caption{\label{fig:framework}
    \textbf{Framework of MeshFlow.}
    First, we represent the mesh as a triangle soup, which shares two levels of permutation invariance.
    To capture the symmetry inside the triangle soup, we build an optimal transport (OT) map between noise $x_0$ and data $x_1$, obtaining the nested noise $\tilde{x}_0$ (Sec.~\ref{subsec:train}). Given the nested coupling ($\tilde{x}_0, x_1$), flow matching builds path with linear interpolating, defining the constant velocity $u_t$ and sample $x_t$.
    In addition, we design an equivariant architecture (Sec.~\ref{subsec:eq-arch}) for modeling the time-dependent velocity field $v_\theta(x_t, t)$ of the triangle soup.
    }
\end{figure*}

%% file: figs/eq-dit.tex
\begin{figure}
    \centering
    \includegraphics[width=1\linewidth]{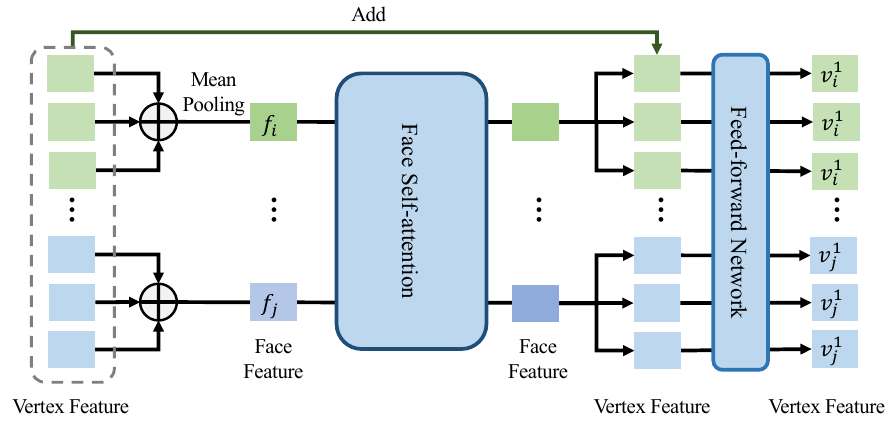}
    \caption{\label{fig:eq-dit}\textbf{Equivariant DiT block}. In consideration of simplicity, we neglect the adaLN block with conditional information (timestamp).
    The DiT block first takes in set of vertex features $\{v_i^1, v_i^2,v_i^3\}_{i=1}^N$. 
    Then the vertex feature $\{v_i^1, v_i^2,v_i^3\}$ in each face is grouped into one face feature $f_i$ by mean pooling.
    Face features $\{f_1, \cdots, f_N\}$ are processed by self-attention.
    Then we add the face feature back to the original vertex feature, ensuring vertex-level equivariance.
    Finally, the vertex features are independently transformed by a feed-forward network to enhance the representation ability.
    }
\end{figure}

%% file: figs/illustrate-ot.tex
\begin{figure}
    \centering
    \includegraphics[width=1\linewidth]{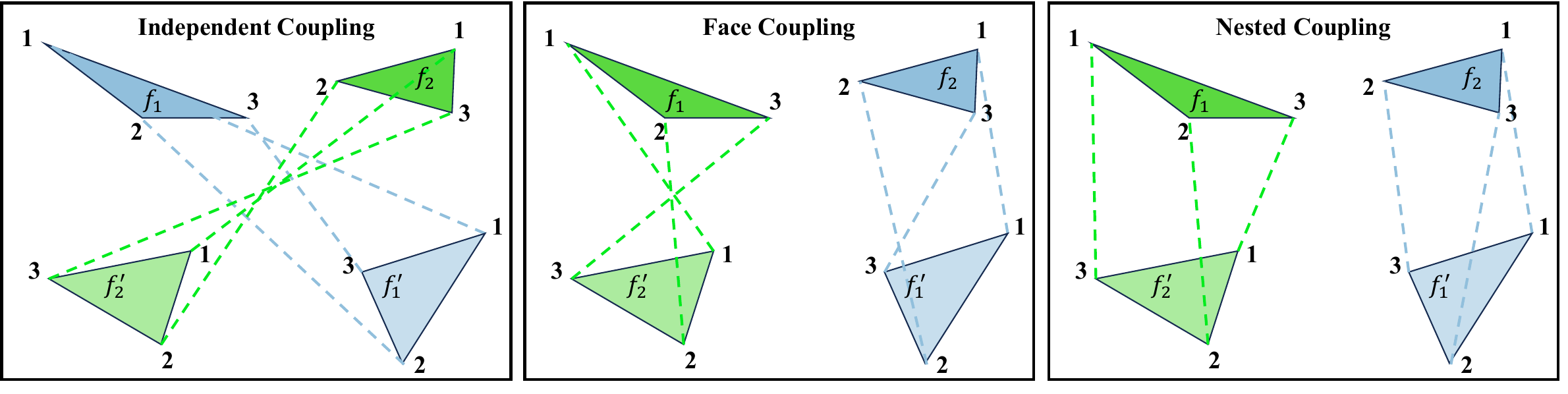}
    \caption{\textbf{2D Coupling Comparison.} Two darker triangles on the top are coupled with two lighter triangles on the bottom using different strategies. Color indicates matched triangles and dotted lines indicate matched vertices. Note that nested coupling results in significantly fewer path intersections compared to face coupling and independent coupling. While face coupling correctly couples the triangles, it still results in vertex crossings. In contrast, our approach results in no crossings and achieves the least cost in Eq.~\ref{eq:our_cost}.}
    \label{fig:ot}
\end{figure}

%% file: secs/5_experiments.tex
\input{tabs/uncond}

\section{Experiments}
\label{sec:Experiments}



\paragraph{Dataset.} 
Following prior works~\cite{chen2024meshxl, siddiqui2024meshgpt}, we evaluate our method on four ShapeNet~\cite{chang2015shapenet} categories: Table, Chair, Lamp, and Bench.
We use the dataset split in MeshXL~\cite{chen2024meshxl}.
Each mesh is normalized to [-0.95, 0.95]$^3$.
To obtain meshes of similar shape but with diverse face counts, we employ quadric edge collapse decimation~\cite{qed}.
The augmented meshes are filtered based on a pre-set maximum Hausdorff distance to maintain fidelity. 
After augmentation, we obtain 126788/80569/14831/13142 meshes for each category, respectively.
We also transform the vertices to have a unit standard deviation and a zero mean. 

\paragraph{Metrics.}
We evaluate MeshFlow from two perspectives: distribution similarity and topological quality.
For distribution assessment, we adopt 1-Nearest Neighbor Accuracy (1-NNA), which measures both fidelity and diversity.
A 1-NNA value approaching 50\% indicates that the generated distribution is indistinguishable from the reference.
Consistent with prior works~\cite{zeng2022lion, pointflow}, we uniformly sample 2,048 points from each mesh and compute 1-NNA using Chamfer Distance (CD) on equal-sized generated and reference sets.
Regarding topological quality, we propose the \textit{Intersected Face Proportion} ($R_{i}$): \( R_{i} = {N_{fi}}/{N} \), where \( N_{fi} \) denotes the number of faces involved in self-intersections.

\paragraph{Baselines.} 
We compare our method with three state-of-the-art autoregressive mesh generative models: MeshXL~\cite{chen2024meshxl}, MeshGPT~\cite{siddiqui2024meshgpt}, and PolyGen~\cite{nash2020polygen}. 
We use their publicly released checkpoints for all experiments. 
We do not report MeshGPT in Lamp/Bench class since model checkpoints are not publicly available. 
Note that MeshXL is pre-trained on a large-scale mesh dataset before finetuning on the ShapeNet categories.
Such large-scale pretraining potentially boosts their performance.
We also compare our own implementation of PolyDiff~\cite{alliegro2023polydiff}, due to the lack of a publicly available model.

\subsection{Unconditional Mesh Generation}
\label{sub:Results II}


\begin{table}[]
    \centering
    \caption{\textbf{Inference Efficiency (in seconds).}
    Our method significantly outperforms autoregressive baselines with an $18\times$ speedup. The proposed post-processing adds minimal latency while ensuring mesh quality.}
    \resizebox{0.8\linewidth}{!}{
    \begin{tabular}{cccccc}
    \toprule
        MeshGPT & MeshXL & PolyDiff & Ours &  + Post-processing\\
        \midrule
        16.271 & 28.931 & 8.528 & \textbf{0.877} & +0.0233 \\
    \bottomrule
    \end{tabular}
    }
    \label{tab:eff}
\end{table}

\Cref{tab:all-in-one} presents the quantitative results for mesh generation.
Our method achieves the best 1-NNA score in 3 out of 4 categories, indicating that our equivariant architecture effectively learns a rich prior for mesh generation.
Compared to PolyDiff~\cite{alliegro2023polydiff}, the state-of-the-art diffusion-based mesh generative model, our approach demonstrates both superior generative performance and fewer face intersections. This highlights the importance of accounting for permutation invariance in meshes when applying diffusion-based generation.
Our method is also more efficient during inference. 
\Cref{tab:eff} shows the average inference speed of different baselines. 
We achieve a speedup of 18.55$\times$ compared to autoregressive methods.
Qualitative results are shown in ~\Cref{fig:comparison-sota}. 
Compared with baselines, we can output meshes with fewer missing and intersecting faces.    
These results suggest that our approach can achieve mesh generation results on par with state-of-the-arts with significantly faster speed along with significant improvements compared to diffusion-based generative work.

\begin{figure*}[ht]
    \centering
    \includegraphics[width=0.9\linewidth]{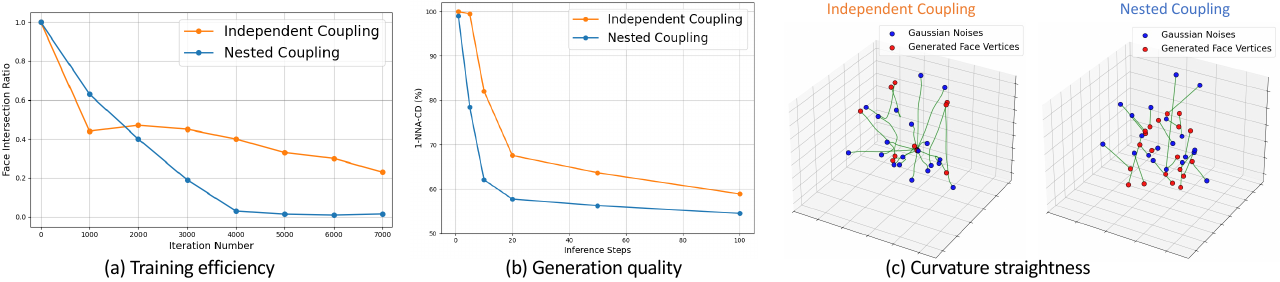}
    \caption{\textbf{Analysis of Nested Optimal Transport.}
    Compared to independent coupling baseline, our nested OT achieves faster training convergence (a); better performance especially in   steps (b); and straighter integral path (c).}
    \label{fig:ot-inf-vis}
\end{figure*}

\begin{figure}[ht]
    \centering
    \includegraphics[width=0.9\linewidth,]{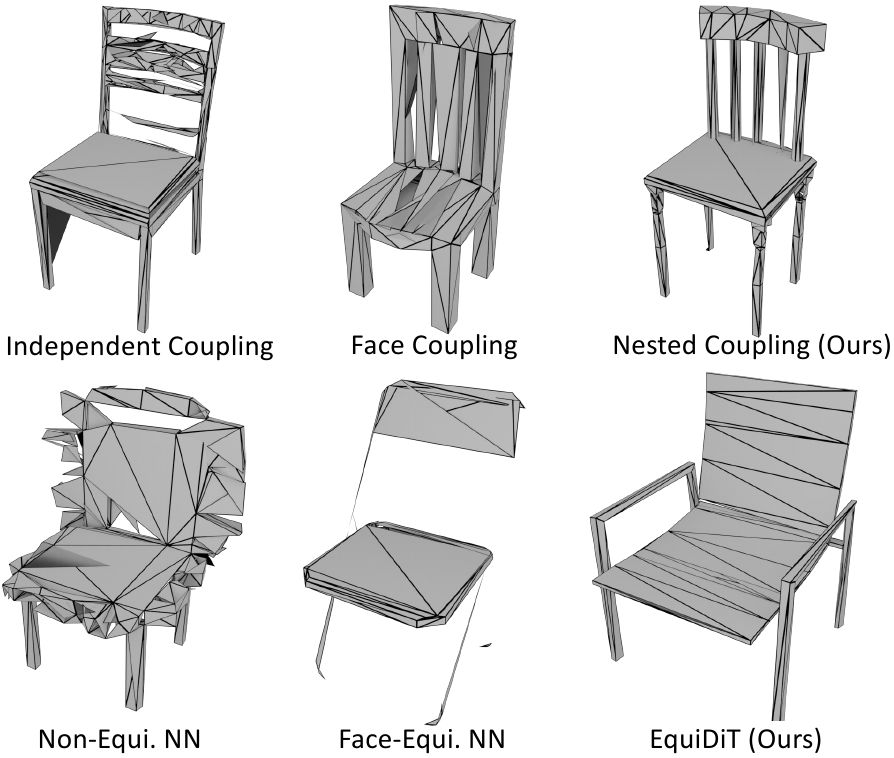}
    \caption{\textbf{Qualitative results for ablative study.} Comparison between different data coupling (top row); comparison between different network architecture (bottom row).
    }
    \label{fig:nest-ot}
\end{figure}

\begin{figure}[ht]
    \centering
    \includegraphics[width=0.9\linewidth]{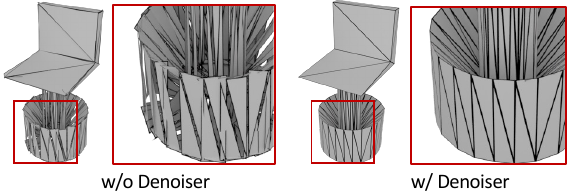}
    \caption{\textbf{Impact of denoiser.} This learnable post-processing effectively removes the low-level noise in the raw model output.}
    \label{fig:denoiser}
\end{figure}

\paragraph{Ablation Study}\label{sec:ablation}
In this section, we validate the effectiveness of two key design choices of our model: the equivariant architecture and the nested OT flow-matching objectives.

\paragraph{Effectiveness of equivariant architecture.}
We validate the effectiveness of the Equivariant DiT Block (\Cref{subsec:eq-arch}) by comparing our method with two baseline variants.
The first, \textit{Non-equi. NN}, uses vanilla DiTs~\cite{peebles2023scalablediffusionmodelstransformers} with positional encodings from the original Transformer~\cite{vaswani2017attention} applied to face features.
The second, \textit{Face-equi. NN}, applies the DiT block to face features obtained via mean pooling over vertex embeddings.
Although this architecture is equivariant to permutations of triangle faces, it lacks equivariance with respect to permutations of vertices within each face.
\qs{
Third, \textit{Vertex-equi. NN}, directly applies the DiT block over vertex embeddings, which ignores the hierarchical grouping of vertices into faces.}
As shown in \Cref{tab:ablation} and \Cref{fig:nest-ot}, our architecture achieves superior 1-NNA performance and produces the most visually coherent results among the three.

\paragraph{Effectiveness of nested optimal transport.}
To validate the effectiveness of our nested optimal transport approach proposed in Sec.~\ref{subsec:train}, we compare it against two alternative coupling methods.
Specifically, we consider: (1) \textit{Independent Coupling (IC)}, which uses the standard flow-matching loss as supervision; and (2) \textit{Face Coupling}, which performs optimal matching only over faces.
As shown in \Cref{fig:nest-ot} and \Cref{tab:ablation}, our method achieves the best 1-NNA score among all coupling variants.
Although the IC variant appears to be on par with our approach in terms of the 1-NNA score (in 50-steps), \Cref{fig:ot-inf-vis}(a) reveals that IC training results in slower convergence in topology quality.
\Cref{fig:ot-inf-vis} (b) evaluates the quality of the generation (1-NNA) in varying inference steps. Our method consistently outperforms the baseline, achieving significantly lower 1-NNA scores. 
Furthermore, \Cref{fig:ot-inf-vis}(c) shows that IC produces more curved flow trajectories, potentially requiring more sampling steps to reach comparable performance.
Therefore, our model obtains high-fidelity results in fewer  (20) function evaluations (\Cref{tab:ablation}), while the IC baseline fails.


\paragraph{Effectiveness of post-processing.}
As reported in \cref{tab:all-in-one}, our post-processing algorithm significantly improves topological quality, reducing the average self-intersection rate ($R_i$) by approximately \textbf{56\%} across all categories. 
Crucially, this refinement preserves the generative fidelity, as evidenced by the comparable 1-NNA.
As illustrated in \Cref{fig:denoiser}, the denoiser effectively rectifies local geometric artifacts. 
It eliminates severe self-intersections observed in the raw output, resulting in cleaner meshes.
Moreover, \Cref{tab:eff} shows that the post-processing algorithm adds a negligible run-time cost.

\input{tabs/ablation}

%% file: tabs/uncond.tex
\begin{table}[tbp]
    \centering
    \caption{
        \textbf{Quantitative Comparisons.} 
        We report 1-NNA and self-intersection rate $R_i$ ($\%$, $\downarrow$).
        The best results are \textbf{bolded} and the second best are \underline{underlined}.
        $^\dagger$From~\cite{chen2024meshxl}. $^\star$Objaverse pre-trained.
    }
    \label{tab:all-in-one}
    \setlength{\tabcolsep}{1.2pt}
    \resizebox{\linewidth}{!}{
        \begin{tabular}{l cc cc cc cc}
            \toprule
            \multirow{2}{*}{Method} & \multicolumn{2}{c}{Chair} & \multicolumn{2}{c}{Table} & \multicolumn{2}{c}{Bench} & \multicolumn{2}{c}{Lamp} \\
            \cmidrule(lr){2-3} \cmidrule(lr){4-5} \cmidrule(lr){6-7} \cmidrule(lr){8-9}
             & 1-NNA & $R_i$ & 1-NNA & $R_i$ & 1-NNA & $R_i$ & 1-NNA & $R_i$ \\
            \midrule
            
            \multicolumn{9}{l}{\textit{Auto-Regressive Models}} \\
            PolyGen$^{\dagger}$ (99.7M)
            & 81.45 & 59.33 & 66.27 & 64.32 & 79.69 & 50.69 & 75.49 & 16.49 \\
            
            MeshGPT (350M)
            & 55.97 & 43.96 & \underline{57.30} & 44.77 & - & - & - & - \\
            
            {MeshXL$^\star$ (1.3B)}
            & \underline{55.32} & \underline{15.20} & 57.78 & \underline{16.53} & \underline{56.25} & 45.56 & \underline{46.77} & 29.14 \\
            
            \midrule
            
            \multicolumn{9}{l}{\textit{Diffusion Models}} \\
            PolyDiff (132M)
            & 79.91 & 91.29 & 73.25 & 69.76 & 61.49 & 40.46 & 70.81 & 83.99 \\
            
            Ours (124M)
            & \textbf{54.51} & 35.42 & 59.14 & 37.50 & \textbf{54.46} & \underline{17.46} & \textbf{51.61} & \underline{15.32} \\
            
            +Post-processing 
            & 56.77 & \textbf{14.02} & \textbf{57.00} & \textbf{15.13} & 59.82 & \textbf{8.56} & 56.45 & \textbf{7.97} \\
            
            \bottomrule
        \end{tabular}
    }
\end{table}

%% file: tabs/ablation.tex


\begin{table}[]
    \centering
        \caption{\textbf{Ablation Study}. 1-NNA scores on the Chair category show that our design (nested coupling and EquiDiT) yields the best performance and remains robust even with only 20 steps thanks to straighter flow trajectory.
        }
        \resizebox{\linewidth}{!}{
    \begin{tabular}{ccccccc}
    \toprule
     \begin{tabular}{@{}c@{}}Inference \\ Steps\end{tabular} & 
     \begin{tabular}{@{}c@{}}Independent \\ Coupling\end{tabular} & \begin{tabular}{@{}c@{}}Face \\ Coupling\end{tabular}& \begin{tabular}{@{}c@{}}Non-equi. \\ NN\end{tabular}  &  \begin{tabular}{@{}c@{}}Face-equi. \\ NN\end{tabular} & \begin{tabular}{@{}c@{}}Vertex-equi. \\ NN\end{tabular} & Ours\\
    \midrule
    50 & 57.42 & 68.87 & 83.87  & 72.42 &87.43 & \textbf{55.97} \\
    20 & 67.58 & 70.71 & 86.58  & 73.03 & 90.20 & \textbf{57.74} \\

    \bottomrule
\end{tabular}}
\label{tab:ablation}
\end{table}

%% file: secs/6_conclusions.tex
\begin{figure}
    \centering
    \includegraphics[width=0.9\linewidth]{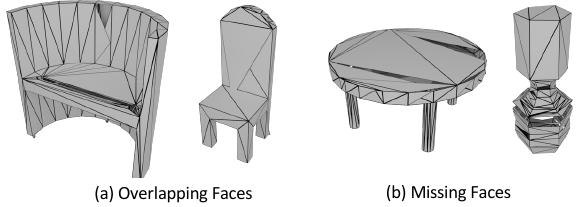}
    \caption{\textbf{Failure cases.} }
    \label{fig:failure}
\end{figure}

\section{Discussion, Limitations, and Future Direction}
\label{sec:Conclusion}

We introduce \textit{MeshFlow}, a novel mesh generative model that leverages equivariant flow matching directly over the triangle-soup representation. 
We identify the key symmetries within triangle soup and design corresponding training objectives as well as a neural network architecture with respect to these symmetries.
Empirically, \textit{MeshFlow} can match performance with state-of-the-art mesh generative models (which are based on autoregressive models) in mesh quality while achieving sub-second inference speed.

\paragraph{Limitation and Future Direction.}
It might seem challenging to scale our coupling algorithm to a large number of faces given its complexity $O(n^3)$.
Toward this end, patch-based training~\cite{jiang2020local} and approximate OT techniques~\cite{bai2023sliced} offer promising avenues to address this limitation.
Our generated meshes occasionally exhibit undesirable artifacts (\Cref{fig:failure}), such as missing or overlapping faces.
We hypothesize that these artifacts stem from limited computational resources and could potentially be solved by large-scale training.
Additionally, extending posterior sampling techniques~\cite{chung2023diffusion} to the equivariant flow matching setting presents an interesting direction for future research.

\begin{acks}
We thank all anonymous reviewers and area chairs for their valuable comments.
The work described in this paper was partially supported by a GRF grant from the Research Grants Council (RGC) of the Hong Kong Special Administrative Region, China [Project No.  CityU 11208123].
This project was also supported in part by NSF-2047677 and 2413161.
The authors acknowledge the support of NVIDIA Academic Hardware Grant Program and the Texas Advanced Computing Center (TACC) at The University of Texas at Austin for providing computational resources that have contributed to the research results reported within this paper.
Kiyohiro Nakayama is supported by National Science Foundation Graduate Research Fellowship program. 
Leonidas Guibas is supported by a Vannevar Bush Faculty Fellowship.
\end{acks}

%% file: secs/8_appendix.tex
\clearpage

\etocdepthtag.toc{mtappendix}
\etocsetnexttocdepth{section}
\etocsettocstyle{}{}

\begin{center}
    {\Huge \textbf{Appendix}}
\end{center}

\localtableofcontents

We have released the full data and code at \url{https://github.com/qiisun/MeshFlow} to contribute to the community of mesh diffusion. 

\section{Proof: Equivariance of the Proposed DiT Block}
\begin{proposition}
The proposed Equivariant DiT block is equivariant to the group $G = S_N \times (S_3)^N$, representing the permutation of faces and the independent permutation of vertices within each face.
\end{proposition}

\begin{proof}
Let the input to the block be a tensor $\mathbf{X} \in \mathbb{R}^{N \times 3 \times C}$, representing a set of $N$ faces, where each face consists of 3 vertices with $C$-dimensional features. We denote the feature of the $j$-th vertex in the $i$-th face as $\mathbf{x}_{i,j}$, where $i \in \{1, \dots, N\}$ and $j \in \{1, 2, 3\}$.

We define the group action of $g = (\sigma, \boldsymbol{\pi}) \in G$ on $\mathbf{X}$, where $\sigma \in S_N$ is a permutation of faces and $\boldsymbol{\pi} = \{\pi_i\}_{i=1}^N$ is a set of permutations of vertices within faces, as:
\begin{equation}
    [g \cdot \mathbf{X}]_{i,j} = \mathbf{x}_{\sigma(i), \pi_{\sigma(i)}(j)}
\end{equation}

The Equivariant DiT block $\Phi$ is composed of Mean Pooling ($P$), Self-Attention ($A$), Broadcasting-Addition ($B$), and a point-wise Feed-Forward Network ($F$). We analyze the transformation of features under $g$ at each step.

\textbf{1. Mean Pooling ($P$):} The block computes face features $\mathbf{f}_i = \frac{1}{3} \sum_{k=1}^3 \mathbf{x}_{i,k}$. Since summation is commutative, $\mathbf{f}_i$ is invariant to the vertex permutation $\pi_i$. Under the face permutation $\sigma$, the face features simply permute:
\begin{equation}
    \mathbf{f}'_i = \frac{1}{3} \sum_{k=1}^3 \mathbf{x}_{\sigma(i), \pi_{\sigma(i)}(k)} = \frac{1}{3} \sum_{k=1}^3 \mathbf{x}_{\sigma(i), k} = \mathbf{f}_{\sigma(i)}
\end{equation}

\textbf{2. Self-Attention ($A$):} The self-attention mechanism processes the set $\{\mathbf{f}_i\}$. Since standard self-attention (without positional encoding) is permutation equivariant with respect to the sequence length $N$:
\begin{equation}
    \mathbf{h}'_i = \text{Attention}(\{\mathbf{f}'_j\}_{j=1}^N)_i = \mathbf{h}_{\sigma(i)}
\end{equation}

\textbf{3. Broadcasting and Addition ($B$):} The updated face features are added back to the vertices: $\mathbf{z}_{i,j} = \mathbf{x}_{i,j} + \mathbf{h}_i$. Applying the group action to the input components:
\begin{equation}
    \mathbf{z}'_{i,j} = [g \cdot \mathbf{X}]_{i,j} + \mathbf{h}'_i = \mathbf{x}_{\sigma(i), \pi_{\sigma(i)}(j)} + \mathbf{h}_{\sigma(i)}
\end{equation}
This is equivalent to permuting the output of the addition step by $g$:
\begin{equation}
    [g \cdot \mathbf{Z}]_{i,j} = \mathbf{z}_{\sigma(i), \pi_{\sigma(i)}(j)} = \mathbf{x}_{\sigma(i), \pi_{\sigma(i)}(j)} + \mathbf{h}_{\sigma(i)}
\end{equation}
Thus, $\mathbf{z}' = g \cdot \mathbf{z}$.

\textbf{4. Feed-Forward Network ($F$):} Since $F$ is a point-wise function applied independently to each vertex feature, it commutes with the permutation operators.
\begin{equation}
    \Phi(g \cdot \mathbf{X}) = F(g \cdot \mathbf{Z}) = g \cdot F(\mathbf{Z}) = g \cdot \Phi(\mathbf{X})
\end{equation}

This concludes the proof that the block is equivariant to $G$.
\end{proof}

\qs{
\subsection{Discussion}
According to the proof, standard token-level positional encodings (e.g., RoPE or sequence IDs) adopted by Diffusion Transformers~\cite{peebles2023scalablediffusionmodelstransformers, ma2024sit} break permutation equivariance because they add different biases to tokens based solely on their sequence order. 
Auto-regressive mesh generation, such as iFlame~\cite{wang2025iflame}, relies on standard transformer backbones with positional encodings, which contrasts with our permutation-equivariant design that avoids fixed ordering biases.
}

\section{Extended Results}
\subsection{Implementation Details: Hyper-parameters}
\label{subsec:imple}
We use DiT-B as the backbone, consisting of 12 layer transformers with 12 heads and a hidden dimension of 768, with 124M parameters in total.
The base resolution $N_0$ is set to 32, and the discount factor is set to 0.75.
The frequency of positional encoding is set to $L$=20.
Similar to DiT, we use AdamW~\cite{loshchilov2017decoupled} optimizer with constant learning rate of 2e-4 and 0 weight decay, with batch size set to 256.
We perform exponential moving average (EMA) training with a decay of 0.9999.
We conduct our main experiments on 4$\times$NVIDIA A100 GPU machine for around 3 days, and the code is implemented with PyTorch.
We use flash attention for all Transformer architecture with \verb|bf16| mixed precision to speed up the training process.
We further adopt the following noise shifting strategy in SD3~\cite{Esser2024ScalingRF} to spend more compute on the high noise regions for meshes with more faces.
Specifically, for a mesh with $N$ faces, we apply a noise schedule in the following form:
$t_N(t) =  (\sqrt{N/N_0} \cdot t)/(1 + ({\sqrt{N/N_0} - 1)\cdot t})$.

\qs{
\subsection{Implementation Details: Nested Optimal Transport}
Algorithm~\ref{alg:nested-ot} details the full nested optimal transport procedure. It requires only a standard linear assignment solver (Hungarian algorithm) and an exhaustive enumeration over the six possible vertex permutations per face pair, which makes the method both computationally efficient and straightforward to implement.
For completeness, we provide the classic Kuhn-Munkres algorithm in Algorithm~\ref{alg:hungarian}. In practice, we use the highly optimized \texttt{scipy.optimize} implementation, which follows the same formulation and delivers the required $  O(N^3)  $.
\begin{algorithm}[t]
\caption{\qs{Nested OT Coupling for Triangle Soups}}
\label{alg:nested-ot}
\begin{algorithmic}[1]
\Require Noisy faces \(\{f_k^0\}_{k=1}^N\), clean faces \(\{f_l^1\}_{l=1}^N\) \quad (each \(f \in \mathbb{R}^{3 \times 3}\))
\Ensure Optimal group element \(g^* = ((\sigma_i^*)_{i=1}^N, \rho^*)\)
\Statex
\State \textbf{Step 1: Build face-cost matrix}
\For{\(k = 1\) to \(N\)}
    \For{\(l = 1\) to \(N\)}
        \State \(\sigma_{kl} \gets \arg\min_{\sigma \in S_3} \|f_l^1 - \sigma \cdot f_k^0\|_2^2\)
        \State \(M_{k,l} \gets \|f_l^1 - \sigma_{kl} \cdot f_k^0\|_2^2\)
    \EndFor
\EndFor
\Statex
\State \textbf{Step 2: Face-level assignment}
\State \(\rho^* \gets \textsc{Hungarian}(M)\) \quad // solves \(\rho^* = \arg\min_{\rho \in S_N} \sum_i M_{i,\rho(i)}\)
\Statex
\State \textbf{Step 3: Retrieve vertex permutations}
\For{\(i = 1\) to \(N\)}
    \State \(\sigma_i^* \gets \sigma_{i,\rho^*(i)}\)
\EndFor
\State \Return \((( \sigma_i^* )_{i=1}^N, \rho^*)\)
\end{algorithmic}
\end{algorithm}
}
\begin{algorithm}[t]
\caption{\qs{Hungarian Algorithm (Kuhn-Munkres)}}
\label{alg:hungarian}
\begin{algorithmic}[1]
\Require Cost matrix \(M \in \mathbb{R}^{N \times N}\)
\Ensure Optimal permutation \(\rho^* \in S_N\)
\Statex
\State \textbf{Initialization}
\State \(u_i \gets \min_j M_{i j}\) for \(i = 1 \dots N\) \quad // row duals
\State \(v_j \gets 0\) for \(j = 1 \dots N\) \quad // column duals
\State \(\rho(i) \gets \text{nil}\), \(\sigma(j) \gets \text{nil}\) for all \(i,j\) \quad // matching
\Statex
\While{exists an unmatched row \(i\)}
    \State Find an augmenting path \(P\) from unmatched row \(i\) in the equality subgraph
          (edges where \(M_{ij} = u_i + v_j\)) using BFS on reduced costs
    \State Let \(\delta \gets\) minimum slack along any candidate edge in the search tree
    \State Update duals:
          \[
          u_i \gets u_i + \delta \quad \text{for all rows in } P, \qquad
          v_j \gets v_j - \delta \quad \text{for all columns in } P
          \]
    \State Augment the current matching along path \(P\)
\EndWhile
\State \(\rho^* \gets\) final matching
\State \Return \(\rho^*\)
\end{algorithmic}
\end{algorithm}

\subsection{Simple Post-processing Details}
Since the raw output of our model is a collection of independent triangles (triangle soup) with minor floating-point inconsistencies, we require a welding operation to recover the underlying topological connectivity. Algorithm \ref{alg:blender_style_merge} outlines our efficient vertex welding strategy based on spatial partitioning. First, we construct a $k$-d tree on all input vertices to accelerate spatial queries. We then iterate through the vertices; for each unvisited vertex $\mathbf{v}_i$, we perform a radius search to identify all neighbors within a distance threshold $\epsilon$ (e.g., $10^{-2}$). These neighbors are spatially clustered and mapped to a single canonical index, effectively "collapsing" them into one vertex. Finally, the mesh faces are reconstructed using these new indices. During this process, we explicitly filter out degenerate faces—triangles where two or more vertices have collapsed into the same index—to ensure the geometric validity of the final mesh.
In practice, we use build-in functions in Blender to achieve this.

\begin{algorithm}[h]
\caption{Vertex Merging}
\label{alg:blender_style_merge}
\begin{algorithmic}[1]
\Require 
    Triangle soup $\mathcal{T}_{in}$ with vertices $\mathcal{V}_{in}$;
    Distance threshold $\epsilon$ (e.g., $10^{-2}$).
\Ensure 
    Manifold mesh $(\mathcal{V}_{out}, \mathcal{F}_{out})$.

\State \textbf{Build} a $k$-d tree $\mathcal{K}$ on all vertices $\mathcal{V}_{in}$.
\State Initialize index map $M: \{1 \dots N\} \to \{1 \dots N\}$, initially $M[i] = -1$.
\State Initialize $\mathcal{V}_{out} \leftarrow \emptyset$.
\State $c \leftarrow 0$ \Comment{Counter for welded vertices}

\For{$i = 1$ to $|\mathcal{V}_{in}|$}
    \If{$M[i] \neq -1$} \textbf{continue} \EndIf \Comment{Already merged}
    
    \State \textbf{Query}: Find set of neighbors $\mathcal{N} \leftarrow \mathcal{K}.\text{radius\_search}(\mathbf{v}_i, \epsilon)$.
    
    \State Append $\mathbf{v}_i$ to $\mathcal{V}_{out}$.
    \State $M[i] \leftarrow c$.
    
    \For{each neighbor $j \in \mathcal{N}$}
        \State $M[j] \leftarrow c$ \Comment{Collapse neighbors to current vertex}
    \EndFor
    \State $c \leftarrow c + 1$
\EndFor

\State \textbf{Reconstruct Faces}:
\For{each face $(a, b, c) \in \mathcal{T}_{in}$}
    \State $f' \leftarrow (M[a], M[b], M[c])$
    \If{$M[a], M[b], M[c]$ are distinct}
        \State Append $f'$ to $\mathcal{F}_{out}$
    \EndIf
\EndFor

\State \Return $(\mathcal{V}_{out}, \mathcal{F}_{out})$
\end{algorithmic}
\end{algorithm}

\subsection{Denoiser Implementation Details}
The denoiser is trained using the exact same dataset split as the primary diffusion backbone. 
To synthesize the noisy inputs encountered during inference, we apply data augmentation by injecting Gaussian noise with a standard deviation of $\eta=0.02$ to the ground-truth mesh vertices. This effectively simulates the positional inaccuracies and discretization errors inherent in the flow matching integration path.
We train the model using a batch size of 128 and an initial learning rate of $1 \times 10^{-4}$ with a cosine decay schedule. 
The training duration is determined by monitoring the validation Mean Absolute Error (MAE). Specifically, we find that the model reaches minimum validation loss at approximately {9k} iterations for Bench, {21k} for Lamp, {30k} for Chair, and {63k} for Table categories.

\begin{figure}
    \centering
    \includegraphics[width=0.9\linewidth]{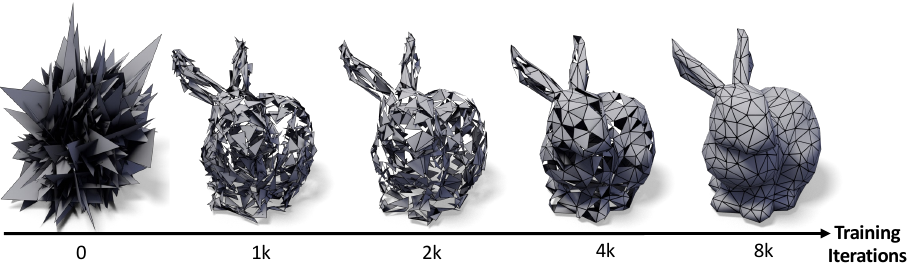}
    \caption{
    \textbf{Topology as an emergent property.} 
    Generated mesh in evolving training iterations.
    }
    \label{fig:bunny}
\end{figure}

\begin{figure}[t]
    \centering
    \includegraphics[width=\linewidth]{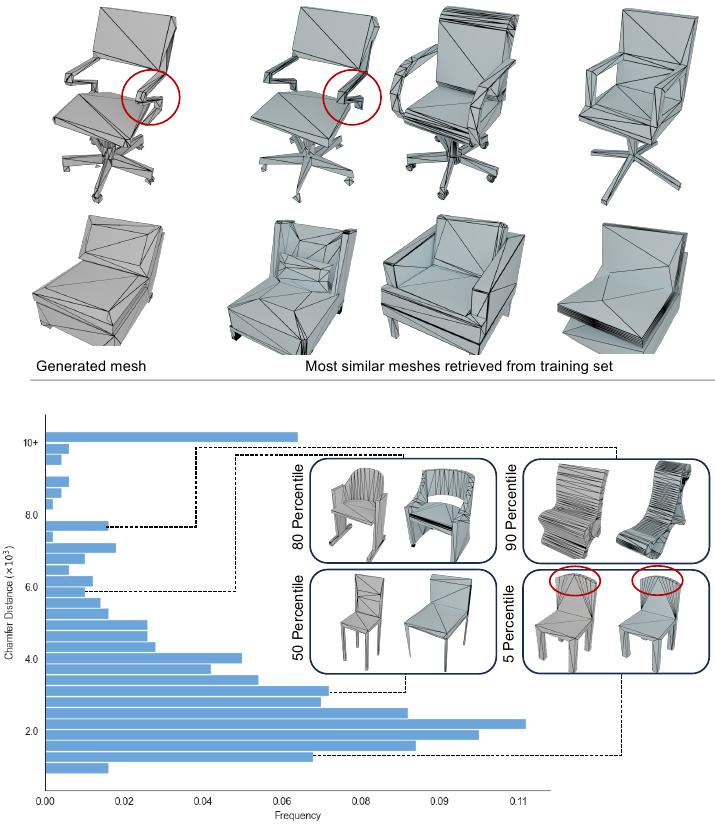}
    \caption{Shape novelty analysis on ShapeNet~\cite{chang2015shapenet} chair category. 
    We show the 3 nearest neighbors in terms of Chamfer Distance (CD) for a generated shape (top). We also plot the distribution of 500 generated chair samples from our method and their closeness to training distribution. Our method can generate shapes that are similar (low CD) as well as different (high CD) from the training distribution, with shapes at the 50th percentile looking different from the closest train shape.}
    \label{fig:diversity}
\end{figure}

\subsection{Convergence of Topology}
Although our triangle soup representation does not explicitly encode mesh connectivity, we observe that coherent topology can emerge as training progresses.
We conduct a MeshFlow training with single bunny data.
Figure~\ref{fig:bunny} illustrates the evolution of a single mesh's geometry and topology quality during training.
As training progresses, we can see that vertex ratio, face intersection ratio, and chamfer distance decrease.
The geometry quality, measured by the Chamfer Distance, improves more significantly in the earlier stage, while the topological quality emerges in the later stage (after 8,000 training steps).

\subsection{Shape Novelty Analysis}
\subsubsection{Long-tailed Distribution}
To demonstrate that our model possesses true generative capabilities rather than merely memorizing the training set, we conduct a shape novelty analysis following previous works~\cite{weng2024pivotmesh, siddiqui2024meshgpt}. 
Specifically, we generate 500 random samples and identify their nearest neighbors from the training set based on the minimum Chamfer Distance (CD). 
The quantitative results are visualized in Figure~\ref{fig:diversity}.

The distribution of minimum CDs exhibits a long-tailed characteristic, indicating a healthy balance between distribution coverage and novelty. 
The lower end of the CD spectrum (e.g., 5th percentile) confirms that our method can generate shapes that faithfully align with the training distribution. 
Crucially, a significant portion of samples falls into the higher CD range (e.g., 50th to 90th percentiles), suggesting that the model is capable of synthesizing highly novel geometries that differ substantially from any training instance.
Furthermore, as shown in the top row of Figure~\ref{fig:diversity}, even for generated shapes with relatively low CD, the retrieved nearest neighbors exhibit distinct structural differences (highlighted in red), proving that our model creates unique variations rather than simply retrieving database copies.

\subsubsection{Similar shapes with different topology}
Figure~\ref{fig:difftopo} demonstrates the stochastic nature of our generative process. 
We display a sequence of generated samples that share a nearly identical visual appearance and geometric hull. 
However, a closer inspection of the wireframes reveals distinct topological structures in each instance. 
This "one-geometry, many-topologies" capability indicates that our model can explore different triangulation strategies for a fixed shape, providing flexibility for downstream applications that may require specific mesh qualities.

\begin{figure}
    \centering
    \includegraphics[width=\linewidth]{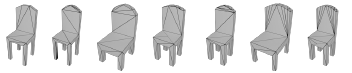}
    \caption{Similar shape with different mesh discretization. }
    \label{fig:difftopo}
\end{figure}

\subsection{Comprehensive Quantitative Evaluation}
\label{subsec:comprehensivequanti}
\subsubsection{Efficiency comparison details.}

Tab.2 is evaluated using the chair category. For each method, we generate 1k meshes and calculate the average inference time on an NVIDIA A6000 GPU with batch-size 1.  We generated meshes with face length randomly sampled from the training set. The autoregressive baselines are run until they generate the EOS token. Complexity-wise, our meshes have on average 452.9 faces while MeshXL’s have 403.5 faces. Our method is both fast and produces meshes with more faces.

\subsubsection{Evaluation details}
In addition to the metrics reported in the main text, we adopt Minimum Matching Distance (MMD) and Coverage (COV) to further analyze the fidelity and diversity of the generated shapes. 
Let $\mathcal{S}_g$ and $\mathcal{S}_r$ denote the set of generated shapes and the reference (test) set, respectively. We employ Chamfer Distance (CD) as the distance measure $D(X, Y)$ between two shapes $X$ and $Y$. The detailed definitions are as follows:
MMD measures the fidelity of the generated samples by calculating the average distance from each shape in the reference set to its nearest neighbor in the generated set. Lower MMD indicates better quality.
\begin{equation}
    \text{MMD}(\mathcal{S}_g, \mathcal{S}_r) = \frac{1}{|\mathcal{S}_r|} \sum_{Y \in \mathcal{S}_r} \min_{X \in \mathcal{S}_g} D(X, Y).
\end{equation}
Coverage measures the diversity of the generated shapes by counting the fraction of reference shapes that are matched to at least one generated shape. Higher COV indicates better coverage of the data distribution.
\begin{equation}
    \text{COV}(\mathcal{S}_g, \mathcal{S}_r) = \frac{|\{\text{argmin}_{Y \in \mathcal{S}_r} D(X, Y) \mid X \in \mathcal{S}_g\}|}{|\mathcal{S}_r|}.
\end{equation}
1-NNA is a classifier-based metric that assesses whether the generated distribution and the reference distribution are distinguishable. The ideal score is $50\%$, indicating that the two distributions are indistinguishable.
\begin{equation}
    \text{1-NNA}(\mathcal{S}_g, \mathcal{S}_r) = \frac{\sum_{X \in \mathcal{S}_g} \mathbb{I}_{X} + \sum_{Y \in \mathcal{S}_r} \mathbb{I}_{Y}}{|\mathcal{S}_g| + |\mathcal{S}_r|},
\end{equation}
where $\mathbb{I}_{X} = \mathbb{I}[N_X \in \mathcal{S}_g]$ and $\mathbb{I}_{Y} = \mathbb{I}[N_Y \in \mathcal{S}_r]$ are indicator functions. Here, $N_X$ is the nearest neighbor of $X$ in the union set $\mathcal{S}_r \cup \mathcal{S}_g \setminus \{X\}$.
Beyond geometric measures, we also evaluate the \textit{perceptual quality} of the mesh surfaces. Since point cloud metrics (e.g., CD-based 1-NNA) may not fully capture surface visual artifacts, we render the generated meshes into shaded images from fixed viewpoints. We then compute Fréchet Inception Distance (FID) and Kernel Inception Distance (KID) on these 2D renderings to quantify the visual similarity between the generated and reference distributions.
To ensure a fair and consistent comparison, we conduct all evaluations using the identical test split of the dataset for both our method and the baselines.
Specifically, for geometric metrics (MMD, COV, 1-NNA), we sample 2,048 points uniformly from the surface of each mesh to compute the Chamfer Distance. Unlike previous works that might use varying splits or sample sizes, we enforce a strict one-to-one correspondence with the official test set to guarantee the validity of our reported results.

\subsubsection{Full quantitative comparison}
In Table~\ref{tab:full-comparison}, our method demonstrates superior generation fidelity, achieving the lowest MMD scores across most categories (e.g., {14.85} on Chair and {25.20} on Lamp), significantly outperforming the large-scale auto-regressive baseline MeshXL (350M) despite using only {124M} parameters. 
In terms of diversity, our method maintains competitive COV scores (e.g., {61.29} on Lamp), indicating that our flow matching framework effectively covers the modes of the data distribution without collapsing.
In Table~\ref{tab:fid-kid-comparison}, consistent with the geometric analysis, our model achieves state-of-the-art perceptual quality on the Chair and Lamp categories (FID {46.57} and {83.07}, respectively), further validating that our generated meshes possess both high-quality geometry and realistic visual appearance.

\begin{table}[h]
    \centering
    \caption{
        \textbf{Perceptual Quality Comparisons on Shaded Images.} 
        We report Fréchet Inception Distance (FID, $\downarrow$) and Kernel Inception Distance (KID, $\times 10^{3}$, $\downarrow$).
        The best results are \textbf{bolded}.
        ``-'' indicates the model does not support the category or results are unavailable.
    }
    \label{tab:fid-kid-comparison}
    \setlength{\tabcolsep}{1.5pt}
        \begin{tabular}{l cc cc cc cc}
            \toprule
            \multirow{2}{*}{Method} & \multicolumn{2}{c}{Chair} & \multicolumn{2}{c}{Table} & \multicolumn{2}{c}{Bench} & \multicolumn{2}{c}{Lamp} \\
            \cmidrule(lr){2-3} \cmidrule(lr){4-5} \cmidrule(lr){6-7} \cmidrule(lr){8-9}
             & FID & KID & FID & KID & FID & KID & FID & KID \\
            \midrule
            
            MeshGPT
            & 76.58 & 64.70 & 84.10 & 53.11 & 80.93 & 20.22 & 174.69 & 39.01 \\
            
            MeshXL
            & 59.38 & 43.37 & \textbf{46.93} & 39.82 & \textbf{54.61} & \textbf{14.18} & 123.58 & 27.61 \\
            
            \midrule
            
            \textbf{Ours}
            & \textbf{46.57} & \textbf{41.05} & 48.71 & \textbf{36.64} & 60.10 & 16.16 & \textbf{83.07} & \textbf{12.25} \\
            
            \bottomrule
        \end{tabular}
\end{table}

\begin{table}[h]
    \centering
    \caption{Ablation study on the effectiveness of the TimeShift}
    \label{tab:ablation_timeshift}
    \begin{tabular}{lcccc}
        \toprule
        \textbf{Method} & \textbf{COV} (\%, $\uparrow$) & \textbf{MMD} ($\downarrow$) & \textbf{1-NNA} (\%) & \textbf{JSD} ($\downarrow$) \\
        \midrule
        w/o TimeShift & 49.35 & 16.50 & 55.81 & \textbf{16.50} \\
        w/ TimeShift & \textbf{49.93}  &  \textbf{14.85} & \textbf{54.51} & 16.51 \\
        \bottomrule
    \end{tabular}
\end{table}

\subsection{Number of faces control}
\begin{figure}
    \centering
    \includegraphics[width=\linewidth]{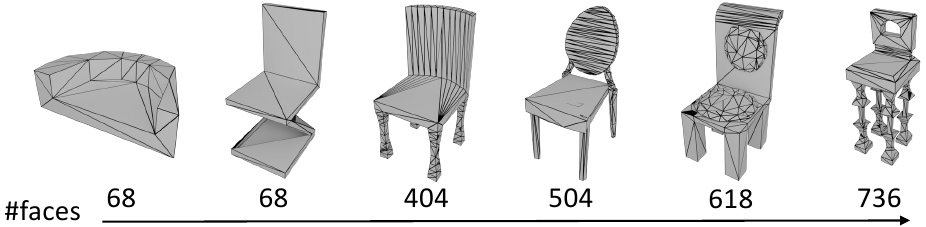}
    \caption{\textbf{Visual comparison of meshes under different face budgets.} Consistent with our quantitative analysis, a \textbf{high face budget} (e.g., 736) yields shapes with fine geometric details and higher curvature. Conversely, a \textbf{low face budget} (e.g., 68) results in a stylistic ``low-poly'' abstraction by smoothing out high-frequency details and producing larger planar regions.}
    \label{fig:numberfaces}
\end{figure}

\begin{table*}[ht]
    \centering
    \caption{
        \textbf{Quantitative Comparisons with Prior Arts on ShapeNet~\cite{chang2015shapenet}. }
        We scale MMD, JSD by $10^3$.
        Our method can produce diverse and high-quality 3D meshes. $\phantom{}^{\dagger}$Metrics for PolyGen are copied from~\cite{chen2024meshxl}. $\phantom{}^\star$Methods are pre-trained on the large-scale Objaverse~\cite{objaverseXL} before being fine-tuned on the specific categories.
    }
    \label{tab:full-comparison}
    \resizebox{\linewidth}{!}{
    \begin{tabular}{clccccclccccc}
    \toprule
       \multirow{3}{*}{Type}    & \multirow{3}{*}{Method (\#Params)} & \multicolumn{5}{c}{Chair} & \multirow{3}{*}{Method (\#Params)} & \multicolumn{5}{c}{Table}\\
    \cmidrule{3-7}\cmidrule{9-13}
      &  &  COV$\uparrow$ & MMD$\downarrow$ & 1-NNA & JSD$\downarrow$ & $R_i (\%) \downarrow$ & &COV$\uparrow$ & MMD$\downarrow$ & 1-NNA & JSD$\downarrow$ & $R_i (\%) \downarrow$ \\ 
      \midrule
    
    \multirow{3}{*}{AR} 
    & PolyGen$^{\dagger}$ (99.7M) & 29.47 & {16.34}& 81.45 & 228.80 & 59.33 & PolyGen$^{\dagger}$ (99.7M)  & 38.67 &  15.84 & 66.27 & 25.06& 64.32\\
    & MeshGPT (350M) & 51.29 & 18.52 & 55.97 & 12.78 & 43.96 & MeshGPT (350M) & 50.77 & 18.25 & 57.30 & \textbf{7.85}  & 44.77 \\
    & MeshXL$^\star$ (350M) & \textbf{52.22} & {17.50} & {55.32} & \textbf{12.29} & {15.20} & MeshXL$^\star$ (350M) & \textbf{52.91} & 16.56 & 57.78 & 9.16 & {16.53}\\
    \midrule
    \multirow{4}{*}{Diffusion} 
    & PolyDiff (132M) & 19.35 & 23.46 & 79.91 & 47.51 & 91.29 & PolyDiff (132M) & 40.46 & 22.21  &  73.25 &  36.34  & 69.76\\
    & Ours (124M) & 49.93 & \textbf{14.85} & \textbf{54.51} & 16.51 & 35.42   & Ours (124M) & 45.13 & \textbf{14.92} & 59.14 & 14.10 & 37.50\\
    & + Post-processing & 50.32 & 16.70 & 56.77 & 13.46  &  \textbf{14.02} & + Post-processing & 46.11 & 15.20 & \textbf{57.00}  & 16.20 & \textbf{15.13} \\

    \midrule
    \multirow{3}{*}{Type} & \multirow{3}{*}{Method (\#Params)} & \multicolumn{5}{c}{Bench} & \multirow{3}{*}{Method (\#Params)} & \multicolumn{5}{c}{Lamp}\\
    \cmidrule{3-7}\cmidrule{9-13}
    &  &  COV$\uparrow$ & MMD$\downarrow$ & 1-NNA & JSD$\downarrow$ & $R_i (\%) \downarrow$& &COV$\uparrow$ & MMD$\downarrow$ & 1-NNA & JSD$\downarrow$ & $R_i (\%) \downarrow$\\ 
    \midrule
     \multirow{2}{*}{AR} 
    & PolyGen$^{\dagger}$ (99.7M) & 37.50 & \textbf{10.8} &   79.69 & 55.25& 50.69 & PolyGen$^{\dagger}$ (99.7M) & 31.76 & 33.87 &  75.49& 81.76 &16.49 \\
    & MeshXL$^\star$ (350M) & \textbf{55.35} & 13.91 & 56.25 & \textbf{26.70} & 45.56 & MeshXL$^\star$ (350M) & 54.83 & 31.33 & {46.77} & 51.59 & 29.14\\
    \midrule
    \multirow{4}{*}{Diffusion} 
    & PolyDiff (132M) & 43.68  &31.31 & 61.49 & 32.60 & 40.46 & PolyDiff (132M) & 48.38 & 38.03 & 70.81 & 69.43 &  83.99\\
    & Ours (124M) &  51.79 & 14.33  & \textbf{54.46} & 32.21 & 17.46 & Ours (124M) & \textbf{61.29} & \textbf{25.20} & \textbf{51.61} & \textbf{38.18} & 15.32  \\
    & + Post-processing & 46.43  & 15.70 & 59.82 & 46.9 & \textbf{8.56} & + Post-processing & 50.32 & 26.90 & 56.45  & 48.54 &  \textbf{7.97} \\
    \bottomrule
        \end{tabular}
}
\end{table*}
In this section, we provide further quantitative evidence supporting the correlation between face counts and geometric details. Table~\ref{tab:face_budget_appendix} reports the average face area and total Gaussian curvature of meshes generated under different face budgets. 

As indicated by the results, shapes generated with a \textbf{high face budget} (e.g., 800) exhibit significantly smaller face areas ($0.011$) and higher Gaussian curvature ($840.03$), demonstrating the preservation of fine geometric details. Conversely, a \textbf{low face budget} (e.g., 100) results in larger face areas ($0.058$) and reduced curvature ($119.98$). This confirms that restricting the face count effectively produces a stylistic ``low-poly'' abstraction by smoothing out high-frequency details.
\begin{table}[ht]
    \centering
    \caption{Quantitative comparison of geometric properties under different face budgets. A higher face budget results in higher curvature (more details), whereas a lower budget leads to larger face areas (low-poly style).}
    \label{tab:face_budget_appendix}
    \begin{tabular}{lcc}
        \toprule
        \textbf{Face Budget} & \textbf{Avg. Face Area} & \textbf{Gaussian Curvature}\\
        \midrule
        Low (100)  & 0.058 & 119.98 \\
        High (800) & 0.011 & 840.03 \\
        \bottomrule
    \end{tabular}
\end{table}

\subsection{More mesh generation results}
\subsubsection{More qualitative comparisons}
We show more comparison results in~\Cref{fig:more-comp} (extension of \Cref{fig:comparison-sota}).
Compared to the baselines, our generation framework produces various and high-quality meshes.
We \textbf{do not} compare the results with PivotMesh~\cite{weng2024pivotmesh} since it does not release the checkpoint in ShapeNet categories.
And MeshGPT does not release the checkpoint in lamp and bench categories.

\qs{
\subsubsection{Evaluation on additional ShapeNet categories}
Following prior baselines (MeshGPT and MeshXL), our primary quantitative evaluation uses four ShapeNet categories. To verify generalization to a broader range of 3D shapes, we additionally trained both our method and MeshXL on six categories selected from the top-10 largest ShapeNet classes.
Table~\ref{tab:extra-shapenet} reports 1-NNA ($\downarrow$) on these additional categories. Our method consistently outperforms MeshXL across all six classes, confirming robust generation quality beyond the original limited set.
\begin{table}[t]
\centering
\caption{\qs{1-NNA ($\downarrow$) on six additional ShapeNet categories (lower is better).}}
\label{tab:extra-shapenet}
\resizebox{\linewidth}{!}{
\begin{tabular}{lcccccc}
\toprule
Method & Airplane & Cabinet & Display & Bathtub & Bottle & Loudspeaker \\
\midrule
MeshXL & 65.38 & 65.63 & 58.62 & 52.63 & 62.50 & 75.60 \\
Ours   & \textbf{53.84} & \textbf{48.48} & \textbf{55.08} & \textbf{51.31} & \textbf{52.50} & \textbf{59.48} \\
\bottomrule
\end{tabular}
}
\end{table}
}

\qs{
\subsection{OT complexity and scalability}
The optimal transport (OT) coupling incurs $  O(N^3)  $ complexity during training only, as it constructs supervision trajectories per batch. At inference, generation involves solely EquiDiT forward passes and ODE solving, yielding approximately $  O(N^2)  $ scaling dominated by transformer and integration steps.
Table~\ref{tab:ot-timing} reports per-batch timings (NVIDIA A100 GPU) comparing OT coupling to diffusion forward+backward passes:
\begin{table}[t]
\centering
\caption{\qs{Per-batch training times (ms) on NVIDIA A100 GPU: OT coupling vs. diffusion forward+backward passes.}}
\label{tab:ot-timing}
\begin{tabular}{lcccc}
\toprule
 & \multicolumn{4}{c}{Number of Faces ($N$)} \\
\cmidrule(lr){2-5}
 & 400 & 800 & 1600 & 3200 \\
\midrule
OT Coupling & 17.02 & 81.94 & 374.82 & 2045.18 \\
Diffusion Fwd + Bwd & 210.64 & 395.73 & 802.33 & 1738.15 \\
\bottomrule
\end{tabular}
\end{table}
Even at $  N=1600  $ (double our default), OT remains faster than network passes and constitutes a small fraction of iteration time. While the current Hungarian-based implementation supports meshes up to several thousand faces, scaling to dense production meshes (10k+ faces) motivates future work on efficient OT solvers~\cite{xia2026a, cui2025optical, cui2025multi} for near-linear complexity.
}

\section{Extended Ablation study}

\subsection{More results of denoiser}
We provide additional visual evidence demonstrating the efficacy of our Denoising Mesh Decoder in Figure~\ref{fig:denoiser2}. While the flow matching model successfully captures the global semantic structure, the raw outputs often suffer from high-frequency artifacts due to the unstructured nature of the triangle soup representation.
As highlighted in the zoomed-in regions of Figure~\ref{fig:denoiser2}, these artifacts manifest as severe self-intersections (top row), irregular surface noise on planar regions (middle row), and disconnected or jagged geometry in thin structures like armrests (bottom row).
The denoiser effectively acts as a geometric projection operator, mapping these noisy, non-manifold inputs to clean, high-quality meshes. It resolves face intersections and regularizes the triangulation patterns, resulting in smooth surfaces and sharp edges without altering the underlying semantic identity of the object.

\begin{figure}
    \centering
    \includegraphics[width=\linewidth]{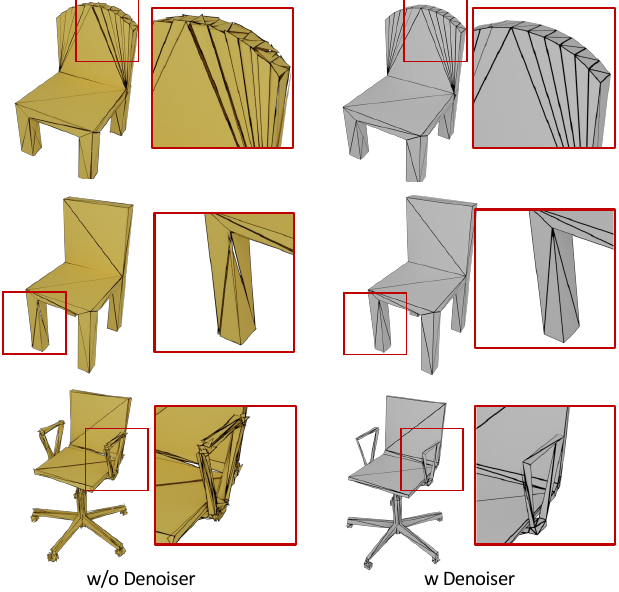}
    \caption{Impact of denoiser. (More cases)}
    \label{fig:denoiser2}
\end{figure}

\subsection{The impact of time shifting}
We conduct an ablation study to investigate the effectiveness of the time-shifting in our flow matching inference. The quantitative results are reported in Table~\ref{tab:ablation_timeshift}.
Comparing the model trained with and without time shifting, we observe a clear performance improvement when the strategy is applied. 
Specifically, the Minimum Matching Distance (MMD) decreases significantly from 16.50 to \textbf{14.85}, indicating that time shifting helps the model generate shapes with higher fidelity. 
Furthermore, the 1-NNA metric improves from 55.81\% to \textbf{54.51\%}, suggesting that the generated distribution aligns more closely with the ground truth. 
The Coverage (COV) also sees a slight increase to \textbf{49.93\%}, while the JSD remains comparable. 
These results confirm that time shifting is a crucial component for stabilizing training and enhancing the overall quality of the generated meshes.
\begin{figure}
    \centering
    \includegraphics[width=\linewidth]{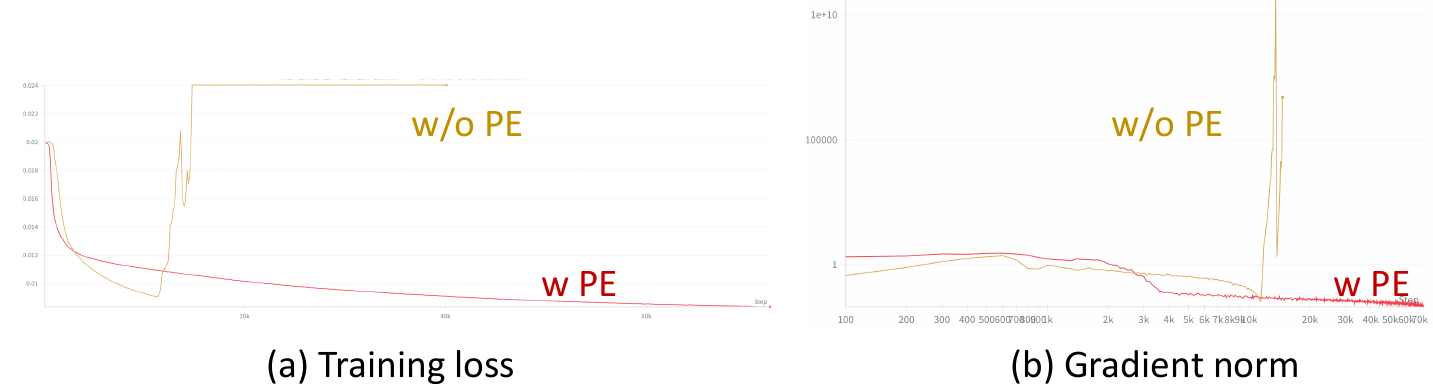}
    \caption{Impact of positional encoding.}
    \label{fig:pe}
\end{figure}

\subsection{The impact of positional encoding}
To validate the necessity of explicit spatial information, we conduct an ablation study by removing the Positional Encoding (PE) module. Figure~\ref{fig:pe} compares the training dynamics of the two settings.
As observed in Figure~\ref{fig:pe}(a), the model trained without PE exhibits severe instability, suffering from a sudden divergence in loss after the initial phase. This is further corroborated by the gradient norm analysis in Figure~\ref{fig:pe}(b), where the removal of PE leads to catastrophic gradient explosion (spiking to $10^{10}$).
In contrast, incorporating PE effectively stabilizes the optimization process, suppressing gradient spikes and ensuring smooth, monotonic convergence. This indicates that PE is critical for the model to correctly distinguish and assemble geometric primitives during the flow matching process.

\begin{figure*}
    \centering
    \includegraphics[width=\linewidth]{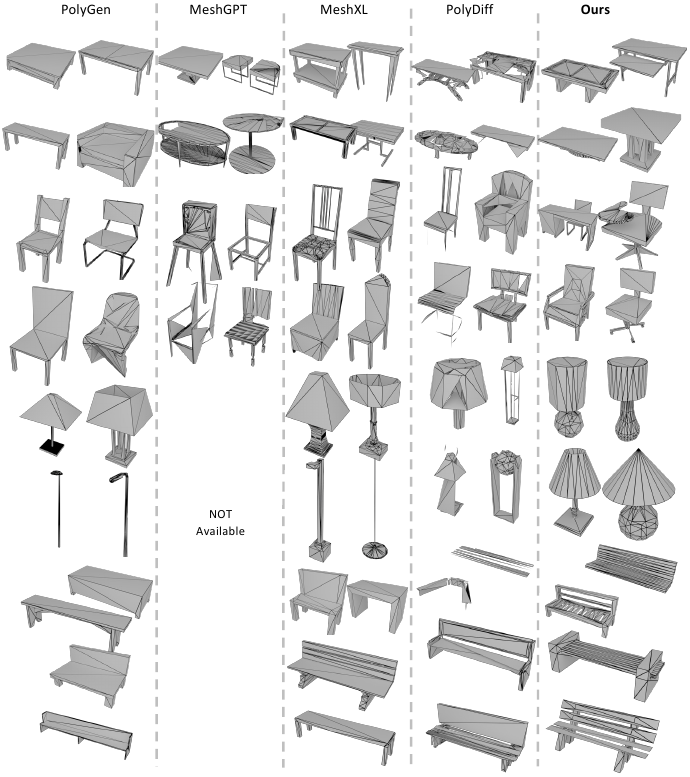}
    \caption{Extended comparison with the state-of-the-arts. We do not compare with MeshGPT in lamp/bench because of the missing checkpoint. We do not compare with PivotMesh since it does not release the checkpoint in shapenet categories,}
    \label{fig:more-comp}
\end{figure*}

\qs{
\section{More Limitations}
While MeshFlow demonstrates promising results, several limitations remain. First, as a proof-of-concept, our current model supports a maximum of 800 faces. It remains challenging to represent product-level meshes (e.g., $>10k$ faces) or to capture highly intricate curvatures and thin structures found in complex real-world data. Future work could explore the scalability of MeshFlow and more efficient network architectures to handle higher resolutions.
Further, our framework currently could not vary the number of generated faces. Our method is designed to generate a mesh given a predefined face number during inference. It is a promising future direction to dynamically predicting a task-aware face budget (e.g., via an auxiliary network that predicts whether a generated face is in the final mesh or not).
}

\qs{
Third, our framework relies on a post-processing stage (incorporating both learnable and traditional algorithms) to derive watertight meshes from raw triangle soups, a design choice primarily dictated by current computational constraints. Eliminating this requirement to achieve a unified, end-to-end generation process represents an elegant and important direction for future research.
}